\documentclass[a4paper,fleqn]{cas-dc}

\usepackage[numbers,compress]{natbib}

\def\tsc#1{\csdef{#1}{\textsc{\lowercase{#1}}\xspace}}
\tsc{WGM}
\tsc{QE}
\tsc{EP}
\tsc{PMS}
\tsc{BEC}
\tsc{DE}

\usepackage{amsmath,amsfonts,amssymb,amsthm}
\usepackage{mathtools}
\usepackage{bbm}
\usepackage{algorithm}
\usepackage{algorithmicx}
\usepackage{algpseudocode}

\usepackage{array}
\usepackage{textcomp}
\usepackage{stfloats}
\usepackage{placeins}
\usepackage{cuted}
\usepackage{url}
\usepackage{verbatim}
\usepackage{graphicx}
\usepackage{adjustbox}

\usepackage{multirow}
\usepackage{booktabs}
\usepackage[utf8]{inputenc}
\usepackage{xcolor}

\usepackage{tikz}
\usetikzlibrary{arrows.meta, positioning, shapes.geometric}
\usepackage{caption}
\usepackage{subcaption}
\makeatletter
\AtBeginDocument{%
  \renewcommand{\figurename}{Fig.}%
  \def\fnum@figure{\textbf{\figurename\ \thefigure.}\@gobble}%
}
\makeatother
\usepackage{pifont}
\usepackage{enumitem}
\newtheorem{definition}{Definition}
\newtheorem{assumption}{Assumption}

\newtheorem{theorem}{Theorem}
\newtheorem{corollary}{Corollary}
\newtheorem{lemma}{Lemma}
\newtheorem{proposition}{Proposition}
\newtheorem{remark}{Remark}

\newcommand{\zc}{\mathbf{z}_c}
\newcommand{\ze}{\mathbf{z}_e}
\newcommand{\zet}{\tilde{\mathbf{z}}_e}

\begin{document}
\let\WriteBookmarks\relax
\setcounter{topnumber}{4}
\setcounter{bottomnumber}{2}
\setcounter{totalnumber}{6}
\setcounter{dbltopnumber}{2}
\renewcommand{\topfraction}{.95}
\renewcommand{\bottomfraction}{.90}
\renewcommand{\textfraction}{.05}
\renewcommand{\floatpagefraction}{.80}
\renewcommand{\dbltopfraction}{.95}
\renewcommand{\dblfloatpagefraction}{.80}
\shorttitle{GenCAR: Risk-Controlled OOD Recommendation}
\shortauthors{Wang et~al.}

\title[mode=title]{GenCAR: Generative Counterfactual Alignment with Risk-Controlled Selection for Out-of-Distribution Recommendation}

\author[1,2]{Qianqian Wang}
\ead{12331103@mail.sustech.edu.cn}

\author[1]{Yunshan Li}
\ead{lyunshan565@gmail.com}

\author[3]{Jiawen Zeng}
\ead{jiawenze@engineering.upenn.edu}

\author[1]{Wenwu Gong}
\ead{gongww@sustech.edu.cn}

\author[1,2]{Lili Yang}\cormark[1]
\ead{yangll@sustech.edu.cn}

\affiliation[1]{organization={Shenzhen Key Laboratory of Safety and Security for Next Generation of Industrial Internet, Southern University of Science and Technology}, 
city={Shenzhen},
country={China}}
\affiliation[2]{organization={Department of Statistics and Data Science, Southern University of Science and Technology},
  city={Shenzhen},
  country={China}}

\affiliation[3]{organization={School of Engineering and Applied Science, University of Pennsylvania},
  city={Philadelphia},
  country={USA}}

\cortext[1]{Corresponding authors}

\begin{abstract}
Serving useful recommendations under distribution shift is crucial for balancing utility and risk in out-of-distribution (OOD) recommendation. However, most existing OOD methods improve ranking or construct counterfactual candidates without controlling the proxy-label false discovery rate (FDR) of the served set. In this work, we formulate OOD serving as the $\alpha$-Valid Counterfactual Recommendation ($\alpha$-VCR) problem to retain candidate support learned from counterfactual supervision while controlling proxy-label FDR, and propose GenCAR, which couples preference-grounded counterfactual supervision with calibrated set selection. In particular, GenCAR fixes the stable-preference representation while intervening on the environmental factor, grounds offline large language model proposals through preference anchors and trust-radius filtering, and uses conformal $p$-values for Benjamini--Hochberg selection. We theoretically bound conditional counterfactual approximation error and prove finite-sample, distribution-free control of proxy-label FDR under exchangeability and positive regression dependence, with a Benjamini--Yekutieli guarantee under arbitrary dependence. Extensive experiments audit realized proxy false discovery proportions and demonstrate that GenCAR consistently enhances OOD candidate recovery across diverse benchmarks.
\end{abstract}

\maketitle
\hypersetup{pdfauthor={Qianqian Wang, Yunshan Li, Jiawen Zeng, Wenwu Gong, Lili Yang}}

\section{Introduction}
\label{sec:intro}
Out-of-distribution (OOD) recommender systems serve users across changing environments, where interaction distributions differ between training and serving~\cite{schnabel2016recommendations,wang2024distributionally}. Due to changes in exposure, item popularity, and temporal context, ranking patterns learned from logged interactions may not transfer reliably to the serving environment~\cite{schnabel2016recommendations,zhang2023invcf,yang2023generic}. In such settings, a naive strategy is to enlarge the top-ranked set served to each user, increasing the chance that it contains useful post-shift items. While this strategy can retain more useful items, it can also admit low-quality candidates and lead to recommendation failures as the set grows~\cite{angelopoulos2023reliable}. Thus, it is essential to determine which and how many candidates to serve for each user, which highlights the utility-risk trade-off in OOD recommendation.

Recent research has proposed OOD recommenders that separate preference-related signals from environmental factors to improve ranking after distribution shift~\cite{zheng2021disentangling,zhang2023invcf,wang2023cor}. Distributionally robust objectives and causal diffusion extend this direction to broader forms of shift~\cite{yang2023generic,wang2024distributionally,zhao2025causaldiffrec}. Yet, these methods are learned from interactions logged in the training environment, which do not directly reveal which items remain useful in the serving environment. To obtain supervision for this unobserved setting, a natural strategy is to construct item-level counterfactuals under an environmental intervention. Generative recommenders provide a practical mechanism for constructing such supervision by modeling interaction distributions, generating oracle-like items, or augmenting training sequences~\cite{wang2023diffrec,yang2023dreamrec,liu2023diffaug,liu2025preferdiff,deldjoo2024genrecsys}. However, their objectives optimize reconstruction or personalized ranking without specifying how stable preference should constrain generation under the intervention. Generated candidates can retain environment-specific patterns and produce recommendation failures after the shift.

Counterfactual generation defines the candidate family, but serving still requires deciding which candidates to retain. Conformal recommendation controls the false discovery rate (FDR), the expected proportion of incorrect items among selected recommendations, for sets returned by a given ranker~\cite{angelopoulos2023reliable}. Recent work further controls unwanted recommendations and expands selected sets using observed user feedback~\cite{detoni2025flowers}. These guarantees assume a fixed candidate family and predefined relevance or unwanted-content labels. In our setting, candidates are generated under an environmental intervention, and their alignment with stable preference is represented by proxy labels. Existing methods do not jointly construct such candidates and control their proxy-label FDR during serving. This limitation raises a pivotal issue:
\begin{center}
\emph{How can we retain counterfactual candidates while controlling proxy-label FDR in the served set?}
\end{center}

In this work, we propose Generative Counterfactual Alignment with Risk-Controlled Selection (GenCAR), a two-stage framework that connects preference-grounded counterfactual supervision with calibrated set selection. We define the serving risk as the pooled FDR of candidate pairs labeled as misaligned by a proxy obtained from an offline large language model (LLM). We formulate OOD serving as the $\alpha$-Valid Counterfactual Recommendation ($\alpha$-VCR) problem (Definition~\ref{def:fdr_aligned}), which seeks to retain counterfactual candidate support while keeping proxy-label FDR below a user-specified level $\alpha$. Overall, GenCAR is preference-grounded, risk-controlled, and LLM-free during online serving, as it transforms offline LLM supervision into $\alpha$-valid served sets through stable-preference constraints and conformal calibration.

Theoretically, we establish finite-sample, distribution-free control of pooled proxy-label FDR under target-environment exchangeability and positive regression dependence on the proxy nulls (Theorem~\ref{thm:fdr}). A Benjamini--Yekutieli (BY) variant extends this guarantee to arbitrary dependence (Corollary~\ref{cor:by_fdr}). We also prove that Benjamini--Hochberg (BH) returns the largest member of its nested step-up family (Lemma~\ref{lem:bh_nested}), showing that calibrated selection retains the largest admissible set within this family. A proxy-transfer bound isolates the error required to relate the proxy-label certificate to causal misalignment (Corollary~\ref{cor:proxy_transfer}). For counterfactual construction, we bound the conditional approximation error introduced by anchor projection, LLM realization, and trust-radius filtering (Proposition~\ref{thm:cf_validity}).

To evaluate the validity and effectiveness of GenCAR, we conduct studies on MovieLens (ML)-100K~\cite{harper2015movielens}, Coat~\cite{schnabel2016recommendations}, and Amazon-Book~\cite{he2016ups}, together with two additional calibration datasets. The results show that realized proxy FDP remains below $\alpha$ on all five evaluated candidate pools and that GenCAR consistently improves OOD candidate recovery. Specifically, compared with the shared causal variational autoencoder (CausalVAE) backbone, GenCAR improves Recall@10 by 11.0\%, 19.1\%, and 43.5\% on the three primary benchmarks. Moreover, GenCAR achieves the highest mean Recall@10 among methods with LLM-free serving on all three benchmarks. Matched ablations show that additional training pairs alone do not explain these gains. On Coat, GenCAR retains sub-millisecond online latency, while inference-time LLM reranking requires approximately 5.4 seconds per user.

Our contributions are summarized as follows:
\begin{itemize}
\item We formulate OOD serving as the \textbf{$\alpha$-VCR problem}, establishing a principled framework to retain counterfactual candidate support while controlling pooled proxy-label FDR below a user-specified level $\alpha$.
\item We propose GenCAR, a two-stage framework that transforms offline LLM proposals into preference-grounded counterfactual supervision and calibrated served sets. GenCAR supports variable-size and empty outputs and remains LLM-free during online serving.
\item We derive a conditional approximation bound for counterfactual construction and establish finite-sample pooled proxy-label FDR control for calibrated selection. We further establish BH-family maximality and quantify the proxy error separating this certificate from causal misalignment.
\end{itemize}

\section{Background}\label{sec:prelim}
In this section, we introduce OOD recommendation under environmental shift, discuss the misalignment risk associated with serving post-shift candidates, and formally define the $\alpha$-Valid Counterfactual Recommendation ($\alpha$-VCR) problem.

\subsection{Preliminaries}
\label{sec:task}
Let $\mathcal U$ and $\mathcal I$ denote the user and item spaces. We use $(u,j)$ for a fixed user--item pair and $(U,J)$ for a random pair. For each user $u\in\mathcal U$, let $x_u$ denote the observed interaction history. The training log
$
\mathcal D^{\mathrm{train}}
=
\{(u,j):y_{uj}=1\}
$ is sampled from $P_{\mathrm{train}}$, while serving pairs follow $P_{\mathrm{test}}$. We consider an OOD setting in which
$P_{\mathrm{train}}\neq P_{\mathrm{test}}$ because the exposure policy, item popularity, or temporal context changes between training and serving~\cite{schnabel2016recommendations,zhang2023invcf,wang2024distributionally}.

To determine which items remain suitable after an environmental change, we model each interaction with a structural causal model (SCM). The latent state contains stable user preference $\zc$, an environmental factor $\ze$, and idiosyncratic noise $\eta$. This factorization separates preference-related information from environment-dependent variation~\cite{scholkopf2021causal,khemakhem2020ivae}. Fig.~\ref{fig:scm} contrasts the resulting OOD recommendation problem with a naive predictor and a disentangled model without cross-environment invariance.

\begin{figure}
\centering
\includegraphics[width=\columnwidth]{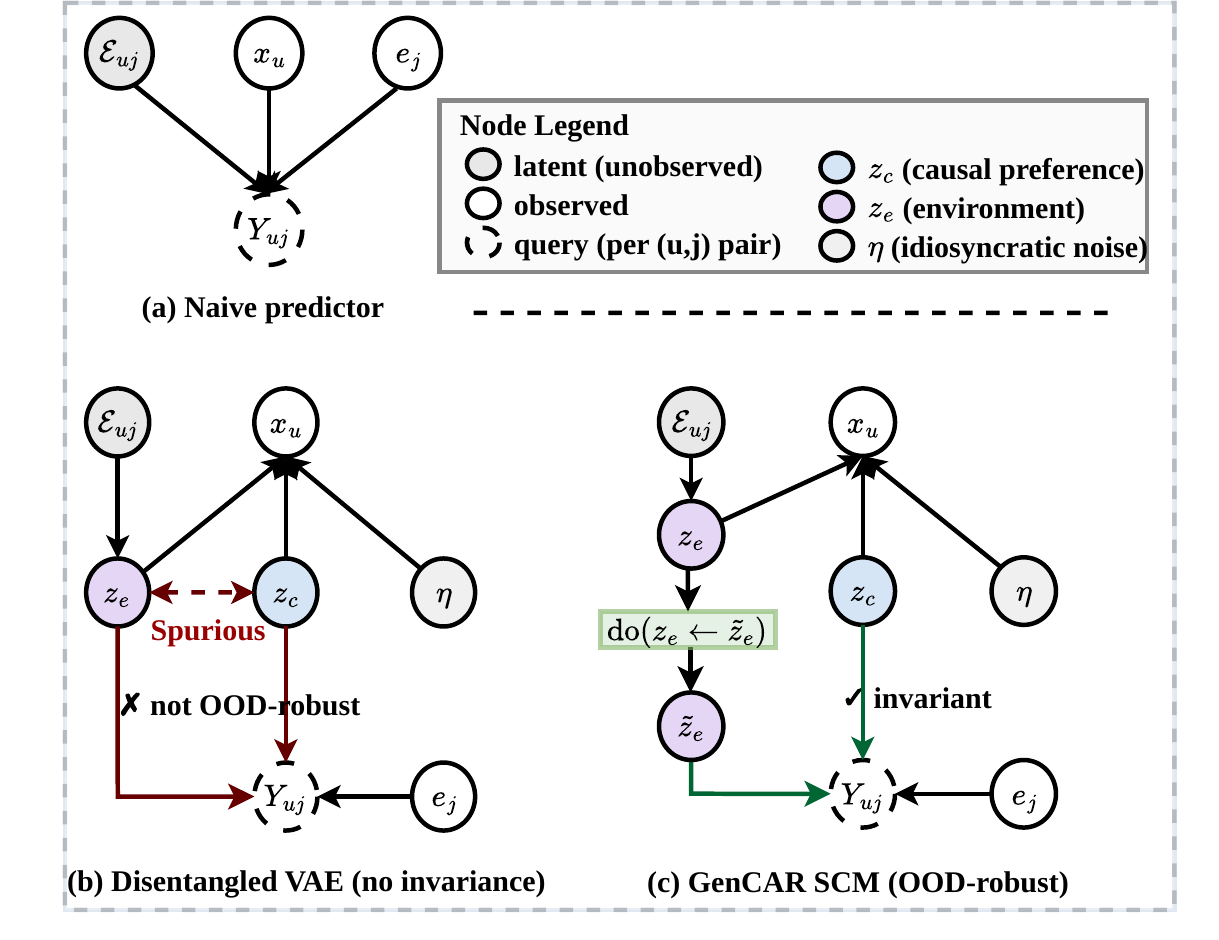}
\caption{Structural causal models for OOD recommendation.
(a) A naive predictor entangles user preference with the observed environment.
(b) A disentangled model separates the latent factors but does not preserve preference across environments.
(c) The target SCM changes the environmental factor $\ze$ while keeping stable preference $\zc$ fixed.}
\label{fig:scm}
\end{figure}

\begin{assumption}[SCM decomposition]\label{asm:scm}
For each user--item pair $(u,j)$, the interaction is generated according to
\[
\zc\sim p(\zc),
\ \
\ze\sim p_e(\ze),
\ \
\eta\sim p(\eta),
\ \
j\sim p(j\mid \zc,\ze,\eta).
\]
Here, $\zc$ represents stable user preference, $\ze$ captures the current environment, and $\eta$ accounts for idiosyncratic variation.
\end{assumption}

\begin{assumption}[Environmental shift]\label{asm:shift}
For each environment $e$, the joint distribution factorizes as
\[
P_e(\zc,\ze,\eta,j)
=
p(\zc)\,p(\eta)\,p_e(\ze)\,
p(j\mid \zc,\ze,\eta).
\]
Across training and serving, $p(\zc)$, $p(\eta)$, and the interaction mechanism
$p(j\mid \zc,\ze,\eta)$ remain invariant \cite{arjovsky2019invariant}, while $p_e(\ze)$ may change.
\end{assumption}

The environmental-shift assumption localizes the distribution change to $\ze$. For a target environment $\tilde{\ze}$, the intervention
$\operatorname{do}(\ze:=\tilde{\ze})$ changes the environment while preserving the stable preference of user $u$. This intervention defines whether an item remains aligned with that user after the shift.

\begin{definition}[Causal alignment target]
For a user--item pair $(u,j)$, define
\[
G_{uj}
=
\mathbbm{1}
\left[
p\!\left(
j
\mid
\zc^{(u)},
\operatorname{do}(\ze:=\tilde{\ze}),
\eta^{(u)}
\right)
\geq \gamma_u
\right],
\]
where $\gamma_u$ is a user-specific alignment threshold. The corresponding aligned item set is
$
\mathcal G_{\tilde{\ze}}(u)
=
\{j\in\mathcal I:G_{uj}=1\}.
$ The set $\mathcal G_{\tilde{\ze}}(u)$ may be empty when no candidate item is aligned with user $u$ in the target environment.
\end{definition}

\noindent\textbf{OOD recommendation.}
A standard OOD recommender relies on a scoring function:
$s_\theta:\mathcal U\times\mathcal I\rightarrow\mathbb R,$
where $s_\theta(u,j)$ represents the predicted relevance of item $j$ to user $u$ under the target environment. We write $\mathcal C_u$ for the top-$K$ ranked items of user $u$ and $\mathcal Q(\mathcal B)=\{(u,j):u\in\mathcal B,\,j\in\mathcal C_u\}$ for the pooled candidate family of a serving batch $\mathcal B$.

Existing OOD recommenders improve $s_\theta$ by separating preference-related signals from environmental factors or by optimizing performance under distribution shift~\cite{zheng2021disentangling,zhang2023invcf,wang2023cor,wang2024distributionally}. Generative recommenders extend this direction by constructing additional supervision or candidate support beyond observed interactions~\cite{wang2023diffrec,yang2023dreamrec,liu2023diffaug,liu2025preferdiff}. These methods can recover useful items that are absent from the training environment. While effective for candidate recovery, the ranking score does not determine whether a candidate belongs to
$\mathcal G_{\tilde{\ze}}(u)$. A natural serving strategy is to retain the highest-ranked candidates. Fixed set sizes and uncalibrated thresholds, however, provide no control over the proportion of retained candidates that are misaligned with the target preference. Retaining more candidates can recover additional post-shift support and can also admit more misaligned recommendations. To this end, we formalize OOD serving as pooled risk control over $\mathcal Q(\mathcal B)$.

\subsection{Problem Formulation}
\label{sec:problem}

Given a ranked candidate family $\mathcal Q(\mathcal B)$, we construct a set-valued selection rule $\Gamma$ that returns a served set as
$\mathcal S_\Gamma
=
\Gamma(\mathcal Q(\mathcal B)).
$ We characterize $\Gamma$ along two dimensions: proxy-label risk and retained candidate support. The causal null set contains candidate pairs that are not aligned with the target preference:
\[
\mathcal H_{0}^{\mathrm{causal}}
=
\{(u,j)\in\mathcal Q(\mathcal B):G_{uj}=0\}.
\]
The indicator $G_{uj}$ is not observed during serving. We therefore construct an operational null set using the LLM-derived alignment proxy:
\[
\mathcal H_{0}^{\mathrm{proxy}}
=
\{(u,j)\in\mathcal Q(\mathcal B):
\text{$(u,j)$ is labeled as misaligned by the proxy}\}.
\]

For a served set $\mathcal S$, define its proxy-label false discovery proportion (FDP) as
\[
\operatorname{FDP}_{\mathrm{proxy}}(\mathcal S)
=
\frac{|\mathcal S\cap\mathcal H_{0}^{\mathrm{proxy}}|}
{\max\{|\mathcal S|,1\}}.
\]
The corresponding pooled false discovery rate is
\[
\operatorname{FDR}_{\mathrm{proxy}}(\mathcal S)
=
\mathbb E
\left[
\operatorname{FDP}_{\mathrm{proxy}}(\mathcal S)
\right],
\]
where the expectation is taken over the target candidate batch and any randomness in the selection rule. The retained candidate support is measured by $\mathbb E[|\mathcal S|]$.

We formulate the goal of retaining candidate support while controlling proxy-label risk as the $\alpha$-Valid Counterfactual Recommendation ($\alpha$-VCR) problem.

\begin{definition}[$\alpha$-Valid Counterfactual Recommendation]
\label{def:fdr_aligned}

Given a risk level $\alpha\in(0,1)$, the goal of $\alpha$-VCR is to find a selection rule $\Gamma_\alpha$ whose served set satisfies the pooled proxy-label FDR constraint:
\begin{equation}
\mathrm{FDR}_{\mathrm{proxy}}
\!\left(\mathcal  S_{\alpha}(\mathcal B)\right)
\leq
\alpha.
\label{eq:alpha_vcr}
\end{equation}

where $\mathcal S_\alpha(\mathcal B)=\Gamma_\alpha(\mathcal Q(\mathcal B))$. Among valid rules, larger $\mathbb E[|\mathcal S_\alpha(\mathcal B)|]$ means that more candidate support is retained.
\end{definition}

By bounding the expected proportion of proxy-misaligned pairs, $\alpha$-VCR prevents candidate expansion from relying on uncontrolled misaligned recommendations. The expected served-set size records how much post-shift candidate support remains under this risk constraint.

In practice, the target proxy-null distribution is unknown, so the $\alpha$-VCR constraint cannot be evaluated exactly. A data-driven solution requires a finite calibration set.

\begin{remark}[Interpretation of validity]
A rule is $\alpha$-valid when its pooled proxy-label FDR does not exceed $\alpha$. An empty set has zero FDP and represents abstention. Corollary~\ref{cor:proxy_transfer} isolates the additional error needed to transfer this proxy certificate to causal misalignment.
\end{remark}

\section{GenCAR}\label{sec:method}

In this section, we present Generative Counterfactual Alignment with Risk-Controlled Selection (GenCAR), a two-stage framework for solving the $\alpha$-VCR problem (Definition~\ref{def:fdr_aligned}). GenCAR first constructs preference-grounded counterfactual supervision under a specified environmental intervention. The accepted counterfactuals train a ranker that forms the post-shift candidate family $\mathcal Q(\mathcal B)$. GenCAR next calibrates evidence against proxy misalignment and selects a variable-size served set satisfying the proxy-label FDR constraint in Eq.~\eqref{eq:alpha_vcr}. All LLM calls and proxy labeling occur offline, leaving online serving to the learned ranker and stored calibration artifacts.

As illustrated in Figure~\ref{fig:framework}, GenCAR contains three modules. Preference-grounded counterfactual construction fixes the stable-preference representation, intervenes on the environmental factor, and converts offline LLM proposals into training supervision. Proxy-label risk calibration fits an alignment predictor and converts a disjoint proxy-null calibration set into conformal evidence. Risk-controlled serving applies BH or BY to the ranked candidate family and returns a pooled served set together with its user-level lists. The first two modules operate offline, while online serving uses only the stored model and calibration artifacts.

\begin{figure*}
    \centering
    \includegraphics[width=\linewidth]{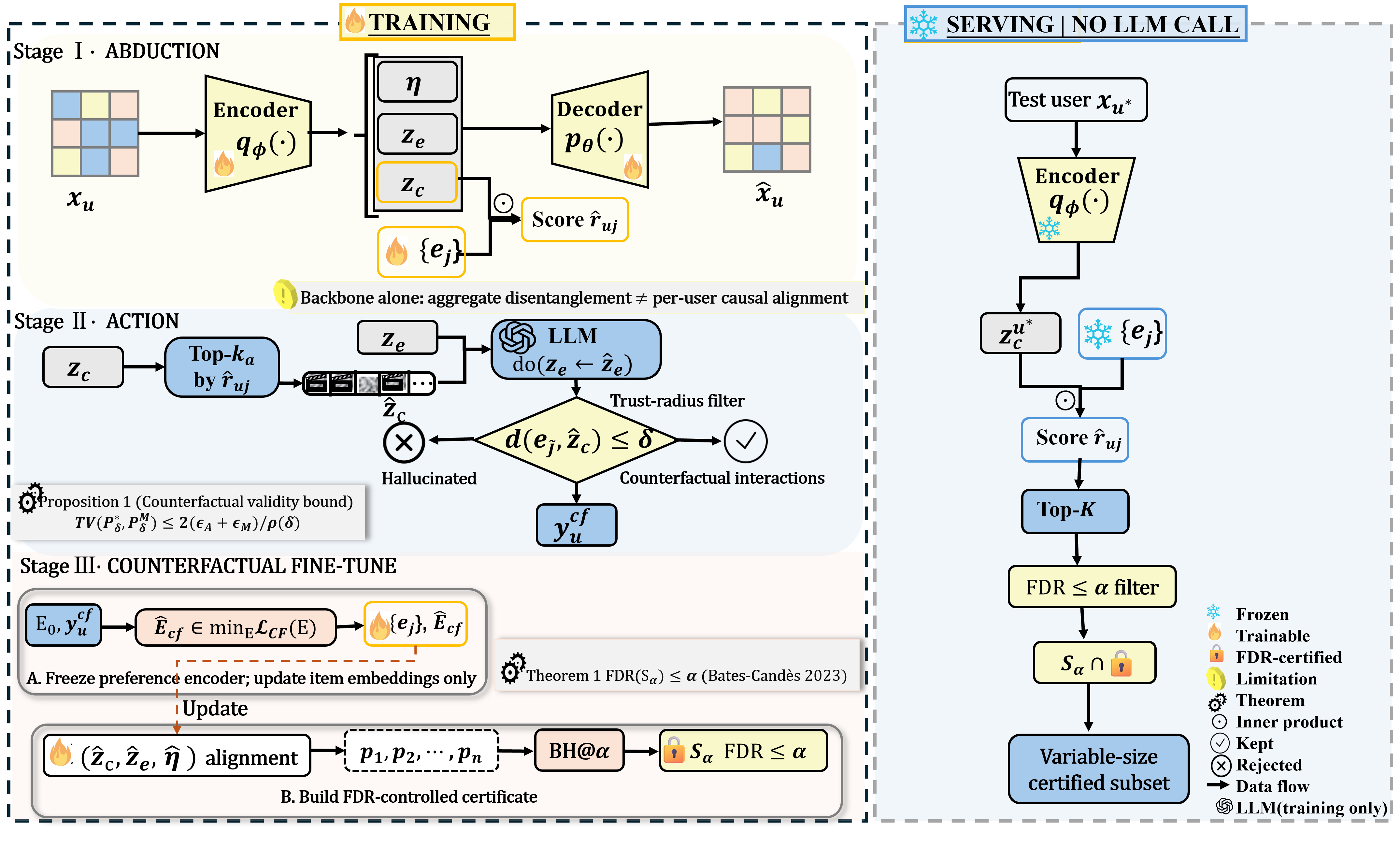}
\caption{
\textbf{Overview of GenCAR.}
GenCAR operates in offline and online phases.
\textbf{Offline construction and calibration (left).}
The backbone separates stable preference from the environmental factor.
GenCAR fixes stable preference, intervenes on the environment, and grounds offline LLM proposals through preference anchors and trust-radius filtering.
The accepted items fine-tune the ranker.
A disjoint proxy-null split supplies the calibration scores used to compute conformal evidence.
\textbf{Risk-controlled serving (right).}
GenCAR forms a ranked candidate family without an LLM call and applies BH or BY at level $\alpha$.
The resulting pooled served set is projected onto user-level lists, with an empty list representing abstention.
}.
\label{fig:framework}
\end{figure*}

\subsection{Preference-Grounded Counterfactual Construction}
\label{sec:counterfactual_construction}

The ranked candidate family $\mathcal Q(\mathcal B)$ must contain items suitable for the target environment, yet these post-shift interactions are absent from the training log. We construct item-level supervision by changing the environmental factor while preserving the learned stable-preference representation.  Given an interaction vector $x_u\in\{0,1\}^{|\mathcal I|}$, the encoder $q_\phi$ produces the factorized posterior
\begin{equation*}
q_\phi(\zc,\ze,\eta\mid x_u)
=
q_\phi(\zc\mid x_u)
q_\phi(\ze\mid x_u)
q_\phi(\eta\mid x_u).
\end{equation*}

The three factors correspond to the SCM in Assumption~\ref{asm:scm}. The stable-preference representation $\zc$ determines ranking, while $\ze$ records environment-dependent variation and $\eta$ captures residual variation.

\noindent\textbf{Preference--environment factorization.} The backbone combines reconstruction, Bayesian personalized ranking (BPR),
and factor separation:
\begin{equation*}
\mathcal L_{\mathrm{backbone}}
=
\mathcal L_{\mathrm{ELBO}}
+
\lambda_{\mathrm{BPR}}\mathcal L_{\mathrm{BPR}}
+
\lambda_{\mathrm{disent}}\mathcal L_{\mathrm{disent}}.
\end{equation*}

Here, $\mathcal L_{\mathrm{ELBO}}$ denotes the negative of the evidence lower bound (ELBO) and
$\lambda_{\mathrm{BPR}},\lambda_{\mathrm{disent}}\geq 0$.
Let $\mathcal T_{\mathrm{obs}}$ denote the observed training triplets,
where $j^+$ is interacted and $j^-$ is unobserved for user $u$.

The ranker scores item $j$ for user $u$ by
$\hat r_{uj}
=
\left\langle \zc^{(u)},e_j\right\rangle $, and optimizes the BPR objective~\cite{rendle2009bpr}
\begin{equation*}
\mathcal L_{\mathrm{BPR}}
=
-\!\!\sum_{(u,j^+,j^-)\in\mathcal T_{\mathrm{obs}}}
\log \sigma\!\left(
\hat r_{uj^+}-\hat r_{uj^-}
\right).
\end{equation*}

The reconstruction and ranking terms preserve predictive information, while $\mathcal L_{\mathrm{disent}}$ separates $\zc$ from $\ze$. Only $\zc$ enters the ranking score $\hat r_{uj}$. This design makes the environmental intervention operational without redefining the learned preference coordinate.

\noindent\textbf{Abduction and environmental intervention.}
We construct counterfactual supervision through
abduction, action, and prediction~\cite{pearl2009causality}.
For each user $u$, the encoder infers
$(\hat\zc^{(u)},\hat\ze^{(u)},\hat\eta^{(u)})$
from $x_u$. We replace $\hat\ze^{(u)}$ with the target factor $\hat\ze$ while keeping $\hat\zc^{(u)}$ and $\hat\eta^{(u)}$ fixed. The resulting counterfactual target is
\begin{equation*}
p\!\left(
j\mid
\hat\zc^{(u)},
\mathrm{do}(\ze:=\hat\ze),
\hat\eta^{(u)}
\right).
\end{equation*}

This distribution describes item suitability in the target environment for the preference state inferred from the observed user history.

\noindent\textbf{Anchor projection and LLM realization.}
The latent vector $\hat\zc^{(u)}$ cannot be supplied directly to an LLM. We map this vector to the catalog items with the highest preference scores:
\begin{equation*}
\mathcal A_u
=
T(\hat\zc^{(u)})
=
\operatorname{TopK}_{j\in\mathcal I}
\left\langle \hat\zc^{(u)},e_j\right\rangle .
\end{equation*}

The item descriptions in $\mathcal A_u$ provide a discrete representation of the learned preference neighborhood. GenCAR combines these anchors with the target-environment descriptor in the offline prompt. The LLM returns catalog item identifiers $\widehat{\mathcal Y}_u$ and an alignment score
$a_{uj}^{\mathrm{LLM}}\in[0,1]$ for each proposal.

\noindent\textbf{Trust-radius filtering.} Let $E_0=\{e_j^0\}_{j\in\mathcal I}$ be the backbone item embeddings used for anchor projection and filtering. Invalid catalog identifiers are removed, and each remaining proposal receives the distance $d_0(\tilde j,\zc^{(u)})=1-\cos(e_{\tilde j}^0,\zc^{(u)})$. The accepted counterfactual set is
\begin{equation}
\mathcal Y_u^{\mathrm{cf}}
=
\left\{
\tilde j\in\widehat{\mathcal Y}_u\cap\mathcal I
:
d_0(\tilde j,\zc^{(u)})\leq\delta
\right\}.
\label{eq:trusted_counterfactuals}
\end{equation}

Equation~\eqref{eq:trusted_counterfactuals} leaves one offline supervision object: catalog-valid proposals inside the preference trust radius.

\noindent\textbf{Counterfactual fine-tuning.} The accepted proposals provide item-level supervision for the target environment. Let $\mathcal U_{\mathrm{cf}}$ denote the users with accepted proposals.
For each $u\in\mathcal U_{\mathrm{cf}}$ with accepted items, $\Pi_u^{\mathrm{cf}}$ pairs $j^+\in\mathcal Y_u^{\mathrm{cf}}$ with an item $j^-$ unobserved for that user. We keep the encoder fixed and update only the item embeddings with the observed and counterfactual BPR losses. GenCAR minimizes
\begin{equation}
\begin{adjustbox}{max width=\linewidth}
$\displaystyle
\mathcal{L}_{\mathrm{CF}}(E)
= -
\sum_{u\in\mathcal{U}_{\mathrm{cf}}}
\mathbb{E}_{(j_+,j_-)\sim\Pi_u^{\mathrm{cf}}}
\left[
\log\sigma\!\left(
\left\langle \hat{\mathbf z}_c^{(u)}, e_{j_+} \right\rangle
-
\left\langle \hat{\mathbf z}_c^{(u)}, e_{j_-} \right\rangle
\right)
\right].
$
\end{adjustbox}
\label{eq:cf_bpr}
\end{equation}

The counterfactual update is defined by $\widehat E_{\mathrm{cf}}\in \arg\min_E \mathcal L_{\mathrm{CF}}(E).$ The encoded preference representations $\hat\zc^{(u)}$ remain fixed, and $E$ is the optimized parameter block in Eq.~\eqref{eq:cf_bpr}. Below, $e_j$ denotes the corresponding updated embedding in
$\widehat E_{\mathrm{cf}}$.

Counterfactual fine-tuning changes item placement while preserving the preference representation used to construct the intervention. The updated ranker forms the candidate family $\mathcal Q(\mathcal B)$. LLM proposals provide offline supervision and are never returned directly during serving. The accepted set $\mathcal Y_u^{\mathrm{cf}}$ supplies ranking supervision, whereas the proxy label introduced in Section~\ref{sec:proxy_calibration} defines the operational null for risk-controlled selection.

\noindent\textbf{Counterfactual approximation path.} The intervention target is not observed directly. We approximate it through anchor projection and LLM realization. For a fixed user and target environment, let $P^\star$, $P^A$, and $P^M$ denote the target, anchor-projected, and LLM-realized distributions:
\begin{equation}
\begin{aligned}
P^\star(j)
&=
p\!\left(
j\mid
\hat\zc^{(u)},
\mathrm{do}(\ze:=\hat\ze),
\hat\eta^{(u)}
\right),\\
P^A(j)
&=
P_M\!\left(
j\mid
T(\hat\zc^{(u)}),
\mathrm{do}(\ze:=\hat\ze)
\right),\\
P^M(j)
&=
P_M\!\left(
j\mid
\operatorname{prompt}
(T(\hat\zc^{(u)}),\hat\ze)
\right).
\end{aligned}
\label{eq:counterfactual_distributions}
\end{equation}

These distributions isolate the approximation introduced by anchor projection and LLM realization. Proposition~\ref{thm:cf_validity} bounds their conditional discrepancy after trust-radius filtering.

\subsection{Proxy-Label Risk Calibration}
\label{sec:proxy_calibration}
Counterfactual fine-tuning determines the ranked candidate family $\mathcal Q(\mathcal B)$, while $\alpha$-VCR requires a separate decision about which pairs enter the served set. The target proxy-null score distribution is unknown, so its FDR constraint cannot be evaluated directly. GenCAR uses a finite, disjoint calibration set to construct conformal evidence against proxy misalignment. This stage returns the pooled served set $\mathcal S_\alpha(\mathcal B)$.

\noindent\textbf{Proxy alignment predictor.} The offline LLM alignment score defines the proxy label
\begin{equation}
Y_{uj}^{\mathrm{proxy}}=\mathbbm{1}\!\left[a_{uj}^{\mathrm{LLM}}\geq\tau\right].
\label{eq:proxy_label}
\end{equation}

Pairs with $Y_{uj}^{\mathrm{proxy}}=0$ instantiate
$\mathcal H_0^{\mathrm{proxy}}$ in Definition~\ref{def:fdr_aligned}. This operational proxy label remains distinct from the unobserved causal indicator $G_{uj}$. GenCAR trains
\begin{equation*}
\hat h:
(\hat\zc^{(u)},\hat\ze^{(u)},e_j)
\longmapsto [0,1]
\end{equation*}
to estimate proxy alignment.
For brevity, we write
$\hat h(u,j):=
\hat h(\hat\zc^{(u)},\hat\ze^{(u)},e_j)$.
The predictor minimizes binary cross-entropy on the alignment split:
\begin{equation}
\begin{aligned}
\mathcal L_h
=
&-\frac{1}{|\mathcal D_{\mathrm{align}}|}
\sum_{(u,j)\in\mathcal D_{\mathrm{align}}}
\Bigl[
Y_{uj}^{\mathrm{proxy}}\log \hat h(u,j)\\
&+
(1-Y_{uj}^{\mathrm{proxy}})
\log(1-\hat h(u,j))
\Bigr].
\end{aligned}
\label{eq:alignment_loss}
\end{equation}
The predictor $\hat h$ transfers the offline proxy labels to user--item pairs that appear in the serving candidate family. The encoder remains fixed while $\hat h$ is trained.

\noindent\textbf{Calibration split.}
Counterfactual fine-tuning, alignment training, and conformal
calibration use disjoint user splits. In particular,
\begin{equation*}
\mathcal D_{\mathrm{proxy}}
=
\mathcal D_{\mathrm{align}}
\cup
\mathcal D_{\mathrm{cal}},
\quad
\mathcal U_{\mathrm{align}}
\cap
\mathcal U_{\mathrm{cal}}
=
\emptyset .
\end{equation*}
Let
\begin{equation*}
\mathcal D_{\mathrm{cal}}^{(0)}
=
\left\{
(u_i,j_i)\in\mathcal D_{\mathrm{cal}}
:
Y_{u_i j_i}^{\mathrm{proxy}}=0
\right\}
\end{equation*}
denote the proxy-null calibration set. No pair in $\mathcal D_{\mathrm{cal}}^{(0)}$ is used to train the encoder, item embeddings, or alignment predictor. Conditional on these learned artifacts, the theoretical guarantee requires the proxy-null calibration scores to be exchangeable with proxy-null scores in the target candidate family.

\noindent\textbf{Conformal evidence.}
For every user--item pair, we define the nonconformity score $V(u,j)=1-\hat h(u,j).$ For a candidate pair $(u,j)$ in the serving batch, its
lower-tail conformal $p$-value is
\begin{equation}
p(u,j)
=
\frac{
1+
\sum_{(u_i,j_i)\in\mathcal D_{\mathrm{cal}}^{(0)}}
\mathbbm{1}[V(u_i,j_i)\leq V(u,j)]
}{
1+|\mathcal D_{\mathrm{cal}}^{(0)}|
}.
\label{eq:method_pvalue}
\end{equation}

A small $p(u,j)$ provides evidence against the proxy-misalignment null. Under conditional target-environment exchangeability, the $p$-value is super-uniform for every pair in $\mathcal H_0^{\mathrm{proxy}}$. Lemma~\ref{lem:pvalue_validity} formalizes this property.

\noindent\textbf{FDR-controlled set selection.}
For a serving batch $\mathcal B$, let $N=|\mathcal Q(\mathcal B)|$ and order the candidate $p$-values as $p_{(1)}\leq\cdots\leq p_{(N)}$. BH select
\begin{equation}
\widehat k_\alpha
=
\max\left\{
k:
p_{(k)}
\leq
\frac{\alpha k}{N}
\right\},
\label{eq:bh_index}
\end{equation}
with $\widehat k_\alpha=0$ when no index satisfies the condition, and returns
\begin{equation}
\mathcal S_\alpha(\mathcal B)
=
\left\{
(u,j)\in\mathcal Q(\mathcal B):
p(u,j)
\leq
\frac{\alpha\widehat k_\alpha}{N}
\right\}.
\label{eq:certified_set}
\end{equation}

The BH rule returns the largest member of its nested step-up family that satisfies the BH condition. Lemma~\ref{lem:bh_nested} establishes this maximality property. Under the exchangeability and dependence conditions of Theorem~\ref{thm:fdr}, the resulting set controls pooled proxy-label FDR at level $\alpha$. When arbitrary dependence is allowed, we replace $\alpha$ with the BY correction $\alpha/H_N$, where $H_N=\sum_{k=1}^{N}k^{-1}$.

\subsection{Risk-Controlled Serving}
\label{sec:serving}

After offline training and calibration, We apply the learned ranker and stored calibration artifacts to each serving batch. The online state is
\begin{equation*}
\Theta_{\mathrm{serve}}
=
\left(
\hat q_\phi,
\hat e_{1:|\mathcal I|},
\hat h,
\{V_i\}_{i\in\mathcal D_{\mathrm{cal}}^{(0)}}
\right).
\end{equation*}
The decoder and LLM are not part of the online path.

\noindent\textbf{Candidate ranking.}
For each user $u$ in a serving batch $\mathcal B$, the encoder infers
$(\hat\zc^{(u)},\hat\ze^{(u)})$ from $x_u$. The counterfactually fine-tuned ranker uses $\hat\zc^{(u)}$ to form $ \mathcal C_u^K=\operatorname{TopK}_{j\in\mathcal I}\left\langle\hat\zc^{(u)},e_j\right\rangle.$ The batch candidate family is
\begin{equation*}
\mathcal Q(\mathcal B)
=
\left\{
(u,j):
u\in\mathcal B,\,
j\in\mathcal C_u^K
\right\}.
\end{equation*}

\noindent\textbf{Certified serving.}
We evaluate $\hat h$ and compute conformal $p$-values for all pairs in $\mathcal Q(\mathcal B)$. Applying Eq.~\eqref{eq:bh_index} produces the pooled served set $\mathcal S_\alpha(\mathcal B)$. The list returned to user $u$ is
\begin{equation*}
\mathcal S_{\alpha}
=
\left\{
j\in\mathcal C_u^K:
(u,j)\in\mathcal S_\alpha(\mathcal B)
\right\},
\end{equation*}
in the original ranker order; an empty list is an abstention. The FDR guarantee applies to the pooled set $\mathcal S_\alpha(\mathcal B)$. Each user-level list inherits the membership decisions of that set. The risk level \(\alpha\) is specified for the pooled serving batch. Each user-level list inherits membership decisions from \(\mathcal S_\alpha(\mathcal B)\); the certificate does not assign a separate FDR guarantee to each user.

Notably, GenCAR offers several compelling advantages:
\begin{itemize}
    \item \textbf{Preference-grounded.}
    GenCAR fixes $\zc$ under the environmental intervention and constrains offline LLM proposals through preference anchors and trust-radius filtering.

    \item \textbf{Risk-controlled.}
    Conformal evidence and BH or BY convert the user-specified level $\alpha$ into pooled proxy-label FDR control under the conditions in Section~\ref{sec:theory}.

    \item \textbf{LLM-free online serving.}
    LLM generation and proxy labeling occur offline. Online serving uses the encoder, item embeddings, alignment predictor, and stored proxy-null scores.
\end{itemize}

Algorithm~\ref{alg:gencar} summarizes the complete offline and online pipeline. The next section bounds the conditional counterfactual approximation error, establishes maximality within the BH family, and proves that calibrated selection satisfies the proxy-label FDR constraint in $\alpha$-VCR.
\begin{algorithm}[t]
\caption{GenCAR training and risk-controlled serving.}
\label{alg:gencar}
\small
\begin{algorithmic}[1]
\Require Training data $\mathcal D^{\mathrm{train}}$,
counterfactual split $\mathcal D_{\mathrm{cf}}$,
alignment split $\mathcal D_{\mathrm{align}}$,
calibration split $\mathcal D_{\mathrm{cal}}$,
offline LLM $\mathcal M$, target intervention $\tilde\ze$,
trust radius $\delta$, anchor size $K_a$, candidate size $K$ and risk level $\alpha$
\Ensure Pooled served set $\mathcal S_\alpha(\mathcal B)$ and user lists $\mathcal S_\alpha(u)$
\Statex \textit{Offline construction}
\State Train the backbone encoder and embeddings $E_0$ on $\mathcal D^{\mathrm{train}}$
\For{each user $u$ represented in $\mathcal D_{\mathrm{cf}}$}
    \State Infer $\zc^{(u)}$ and form preference anchors $\mathcal A_u$
    \State Generate $\widehat{\mathcal Y}_u\gets\mathcal M(\mathcal A_u,\zet)$
    \State Filter proposals by Eq.~\eqref{eq:trusted_counterfactuals} to obtain $\mathcal Y_u^{\mathrm{cf}}$
    \State Form $\Pi_u^{\mathrm{cf}}$ from $j^+\in\mathcal Y_u^{\mathrm{cf}}$ and user-unobserved $j^-$
\EndFor
\State Optimize Eq.~\eqref{eq:cf_bpr} to obtain $E_{\mathrm{cf}}$

\Statex \textit{Offline calibration}
\State Label alignment and calibration pairs by Eq.~\eqref{eq:proxy_label}
\State Fit $h$ on $\mathcal D_{\mathrm{align}}$ and form $\mathcal D_{\mathrm{cal}}^{(0)}$
\State Store $\{V_i=1-h(u_i,j_i)\}_{(u_i,j_i)\in\mathcal D_{\mathrm{cal}}^{(0)}}$

\Statex \textit{LLM-free serving}
\State Infer $\zc^{(u)},\ze^{(u)}$ and rank each $u\in\mathcal B$ into $\mathcal C_u$; pool $\mathcal Q(\mathcal B)$
\State Compute $h(u,j)$ and Eq.~\eqref{eq:method_pvalue} for all $(u,j)\in\mathcal Q(\mathcal B)$
\State Apply Eqs.~\eqref{eq:bh_index}--\eqref{eq:certified_set}
\For{each $u\in\mathcal B$}
    \State Return $\mathcal S_\alpha(u)=\{j\in\mathcal C_u:(u,j)\in\mathcal S_\alpha(\mathcal B)\}$ in ranker order
\EndFor
\end{algorithmic}
\end{algorithm}

\section{Theoretical Analysis}
\label{sec:theory}
In this section, we analyze the construction and selection stages of
GenCAR. We first bound the conditional approximation error introduced
by anchor projection, LLM realization, and trust-radius filtering
(Proposition~\ref{thm:cf_validity}). We next show that the BH candidate
sets form a nested family with a largest step-up-admissible member
(Lemma~\ref{lem:bh_nested}) and establish finite-sample validity of the
proxy-null conformal $p$-values (Lemma~\ref{lem:pvalue_validity}).
Finally, we prove pooled proxy-label FDR control at level $\alpha$
(Theorem~\ref{thm:fdr}), extend the guarantee to arbitrary dependence
(Corollary~\ref{cor:by_fdr}), and provide a proxy-to-causal risk
transfer bound
(Corollary~\ref{cor:proxy_transfer}).

\subsection{Counterfactual Approximation and Selection Structure}
\label{sec:theory_structure}
We first analyze the counterfactual construction stage.
For a fixed user $u$, let
 $\mathcal T_u(\delta)=\{j\in\mathcal I:d_0(j,\zc^{(u)})\leq\delta\}$ denote the trust-radius acceptance event. Recall that $P^\star$, $P^A$, and $P^M$ denote the counterfactual target, anchor-projected distribution, and LLM realization defined in Eq.~\eqref{eq:counterfactual_distributions}. We compare the target and realized distributions after applying the same acceptance event.

\begin{proposition}[Conditional error propagation under trust-radius filtering]
\label{thm:cf_validity}
Assume $P^\star(\mathcal T_u(\delta))>0$ and $P^M(\mathcal T_u(\delta))>0$. Let
\[
P_\delta^\star=P^\star(\cdot\mid\mathcal T_u(\delta)),\qquad
P_\delta^M=P^M(\cdot\mid\mathcal T_u(\delta)),
\]
where $\mathrm{TV}$ denotes total variation distance, and let
$\epsilon_A=\mathrm{TV}(P^\star,P^A)$, $\epsilon_M=\mathrm{TV}(P^A,P^M)$, and $\rho_\delta=\min\{P^\star(\mathcal T_u(\delta)),P^M(\mathcal T_u(\delta))\}$. Then
\begin{equation*}
\mathrm{TV}(P_\delta^\star,P_\delta^M)
\leq
\frac{2(\epsilon_A+\epsilon_M)}{\rho_\delta}.
\end{equation*}
For the causal-misalignment event $\mathcal M_u=\{j\in\mathcal I:G_{uj}=0\}$,
\begin{equation*}
P_\delta^M(\mathcal M_u)
\leq
P_\delta^\star(\mathcal M_u)
+
\frac{2(\epsilon_A+\epsilon_M)}{\rho_\delta}.
\end{equation*}
\end{proposition}

The proof is provided in
Appendix~\ref{app:proof_cf_validity}.
Proposition~\ref{thm:cf_validity} quantifies construction fidelity
after trust-radius filtering. The terms $\epsilon_A$ and
$\epsilon_{M}$ isolate the errors from anchor projection
and LLM realization, while $\rho_\delta$ records the probability mass
retained by the filter. The result links the realized counterfactual
distribution to the target intervention on the accepted region.

We next analyze the structure of the calibrated served sets.
For a serving batch $\mathcal B$, let
$\mathcal Q(\mathcal B)$ contain $N$ candidate pairs.
Write their ordered conformal $p$-values as
$p_{(1)}\leq\cdots\leq p_{(N)}$, and let $\mathcal S_k$ contain the
$k$ pairs with the smallest $p$-values. This family provides the
candidate sets over which the BH rule operates.

\begin{lemma}[Nestedness and step-up maximality]
\label{lem:bh_nested}
The ordered candidate sets form a nested family: $\mathcal S_0 \subseteq\mathcal S_1\subseteq\cdots \subseteq\mathcal S_N.$ For a risk level $\alpha\in(0,1)$, define
\begin{equation*}
\widehat k_\alpha
=
\max\left\{
k:
p_{(k)}
\leq
\frac{\alpha k}{N}
\right\},
\end{equation*}
with $\widehat k_\alpha=0$ when the set is empty.
Then
$\mathcal S_\alpha=\mathcal S_{\widehat k_\alpha}$ is the largest
member of the nested family satisfying the BH step-up condition.
Moreover, for any $0<\alpha_1\leq\alpha_2<1$,
\begin{equation*}
\mathcal S_{\alpha_1}
\subseteq
\mathcal S_{\alpha_2}.
\end{equation*}
\end{lemma}

The proof is given in Appendix~\ref{app:proof_bh_nested}. Lemma~\ref{lem:bh_nested} establishes the size property used by GenCAR. At a fixed level $\alpha$, the method returns the largest set admitted by the BH step-up family. The nesting in $\alpha$ provides an ordered control over served-set size.

\subsection{Proxy-Label Risk Control}
\label{sec:theory_risk}
We now present the main theoretical guarantee of GenCAR, and show that
conformal evidence computed from a finite proxy-null calibration set
controls the expected proxy-misalignment fraction in an unseen serving
batch.

Let $n_0=|\mathcal D_{\mathrm{cal}}^{(0)}|$. The same $V$ and $p$ from Eq.~\eqref{eq:method_pvalue} are used in the algorithm and guarantee.
The first result establishes the finite-sample validity of
this evidence under the proxy-misalignment null.

\begin{lemma}[Finite-sample proxy-null validity]
\label{lem:pvalue_validity}
Suppose that the alignment predictor $\hat h$ and candidate family
$\mathcal Q(\mathcal B)$ are constructed independently of
$\mathcal D_{\mathrm{cal}}^{(0)}$. For every $(u,j)\in\mathcal H_0^{\mathrm{proxy}}$, assume that its
nonconformity score and the $n_0$ proxy-null calibration scores are
exchangeable, conditional on the learned artifacts and candidate
family. The $p$-value is the weak-rank statistic defined in
Eq.~\eqref{eq:method_pvalue}.

Then, for every $t\in[0,1]$,
\begin{equation*}
\Pr\!\left(
p(u,j)\leq t
\mid
\hat h,\mathcal Q(\mathcal B)
\right)
\leq t.
\end{equation*}
Thus, every proxy-null conformal $p$-value is super-uniform.
\end{lemma}

The proof is provided in Appendix~\ref{app:proof_pvalue_validity}. The weak-rank construction in Eq.~\eqref{eq:method_pvalue} is
super-uniform under exchangeability and remains conservative when scores contain ties. Lemma~\ref{lem:bh_nested} and Lemma~\ref{lem:pvalue_validity} provide the two ingredients for calibrated selection. Nestedness identifies the largest BH-admissible set, while proxy-null validity supplies the evidence required for FDR control.

\begin{theorem}[Proxy-label FDR control]
\label{thm:fdr}
Condition on the learned artifacts and the candidate family
$\mathcal Q(\mathcal B)$. Let
\[
N
=
|\mathcal Q(\mathcal B)|,
\quad
m_0
=
|\mathcal Q(\mathcal B)
\cap\mathcal H_0^{\mathrm{proxy}}|.
\]
Assume that the conditions of
Lemma~\ref{lem:pvalue_validity} hold and that the conformal
$p$-value vector is positive regression dependent on a subset (PRDS) of
proxy-null hypotheses, conditional on the candidate family. Then BH
at level $\alpha$ returns
$\mathcal S_\alpha(\mathcal B)$ satisfying
\begin{equation*}
\mathbb E\!\left[
\frac{
|\mathcal S_\alpha(\mathcal B)
\cap\mathcal H_0^{\mathrm{proxy}}|
}{
\max\{|\mathcal S_\alpha(\mathcal B)|,1\}
}
\,\middle|\,
\hat h,\mathcal Q(\mathcal B)
\right]
\leq
\frac{m_0}{N}\alpha
\leq
\alpha.
\end{equation*}
Consequently,
\begin{equation*}
\mathrm{FDR}_{\mathrm{proxy}}
\bigl(\mathcal S_\alpha(\mathcal B)\bigr)
\leq \alpha.
\end{equation*}
\end{theorem}

The proof is given in Appendix~\ref{app:proof_fdr}. Theorem~\ref{thm:fdr} controls the proxy-label risk at any finite
calibration size without specifying a parametric null-score
distribution. Shared-calibration marginal conformal $p$-values
satisfy PRDS under jointly independent and continuously distributed
null scores~\cite{bates2023testing}.

\begin{remark}[Interpretation of the guarantee]
\label{rem:prds}
The certified object in Theorem~\ref{thm:fdr} is the pooled
candidate-pair set $\mathcal S_\alpha(\mathcal B)$. Each user-level
list inherits the membership decisions of this set. The guarantee
controls the expected proxy-misalignment fraction over the serving
batch.
\end{remark}

\begin{corollary}[Arbitrary-dependence control]
\label{cor:by_fdr}
Under the conditions of Lemma~\ref{lem:pvalue_validity}, define
$H_N=\sum_{k=1}^{N}\frac{1}{k}.$ Applying the Benjamini--Yekutieli rule at level
$\alpha/H_N$ yields
\begin{equation*}
\mathrm{FDR}_{\mathrm{proxy}}
\bigl(\mathcal S_{\alpha}^{\mathrm{BY}}(\mathcal B)\bigr)
\leq
\alpha
\end{equation*}
under arbitrary dependence among the candidate $p$-values.
\end{corollary}

The proof is included in Appendix~\ref{app:proof_fdr}.
Corollary~\ref{cor:by_fdr} provides the same proxy-label certificate
without a dependence restriction. GenCAR uses BH under PRDS and BY
under arbitrary dependence.

\begin{corollary}[Proxy-to-causal risk transfer]
\label{cor:proxy_transfer}
Let $\mathcal H_0^{\mathrm{causal}}
=
\left\{
(u,j)\in\mathcal Q(\mathcal B):
G_{uj}=0
\right\}$ denote the causal-misalignment null. Define
\begin{equation*}
\mathrm{FDR}_{\mathrm{causal}}(\mathcal S)
=
\mathbb E\!\left[
\frac{
|\mathcal S\cap\mathcal H_0^{\mathrm{causal}}|
}{
\max\{|\mathcal S|,1\}
}
\right]
\end{equation*}
and
\begin{equation*}
\Delta_{\mathrm{proxy}}(\mathcal S)
=
\mathbb E\!\left[
\frac{
|\mathcal S\cap
(\mathcal H_0^{\mathrm{causal}}
\setminus
\mathcal H_0^{\mathrm{proxy}})|
}{
\max\{|\mathcal S|,1\}
}
\right].
\end{equation*}
The BH set in Theorem~\ref{thm:fdr} satisfies
\begin{equation*}
\mathrm{FDR}_{\mathrm{causal}}
\bigl(\mathcal S_\alpha(\mathcal B)\bigr)
\leq
\alpha
+
\Delta_{\mathrm{proxy}}
\bigl(\mathcal S_\alpha(\mathcal B)\bigr).
\end{equation*}
If
$\mathcal H_0^{\mathrm{causal}}
\subseteq
\mathcal H_0^{\mathrm{proxy}}$,
then
$\Delta_{\mathrm{proxy}}
(\mathcal S_\alpha(\mathcal B))=0$ and
$\mathrm{FDR}_{\mathrm{causal}}
\bigl(\mathcal S_\alpha(\mathcal B)\bigr)\leq\alpha.$
\end{corollary}

The proof is provided in Appendix~\ref{app:proof_proxy_transfer}.
Corollary~\ref{cor:proxy_transfer} isolates the causal nulls missed by the observable proxy. Proposition~\ref{thm:cf_validity} supports the counterfactual construction stage. Theorem~\ref{thm:fdr} answers the $\alpha$-VCR question for proxy-label risk, while Corollary~\ref{cor:proxy_transfer} translates this certificate into a
causal-risk bound through $\Delta_{\mathrm{proxy}}$.


\section{Experiments}
\label{sec:experiments}
In this section, we present the experimental results to validate four claims: (i) calibrated selection keeps
realized proxy FDP below \(\alpha\) on the evaluated diagnostic pools; (ii) preference-grounded counterfactual supervision improves top-10 OOD candidate recovery; (iii) the gains arise from counterfactual content, preference grounding, and trust-radius filtering; and (iv) GenCAR preserves an LLM-free, sub-millisecond serving path under full and partial offline coverage.

\subsection{Experimental Setup}
\label{sec:exp_setup}

\noindent\textbf{Datasets.} We evaluate GenCAR under three common forms of recommendation shift.
ML-100K~\cite{harper2015movielens} represents temporal shift, Coat~\cite{schnabel2016recommendations} represents exposure shift, and Amazon-Book~\cite{he2016ups} represents popularity shift. These
three datasets form the main ranking evaluation. We further include MovieLens-1M~\cite{harper2015movielens} and Amazon-Beauty~\cite{he2016ups}
to evaluate calibration across larger temporal and popularity-shift pools. Steam~\cite{kang2018sasrec} is used to study the setting in which offline LLM generation is available for only a subset of users. Table~\ref{tab:datasets} summarizes the datasets and their shift types.
\begin{table}
		\centering
		\caption{Dataset statistics for the ranking and calibration
evaluations. T, E, and P denote temporal, exposure, and popularity
shifts, respectively.}
		\label{tab:datasets}
		\small
		\setlength{\tabcolsep}{2pt}
		\begin{tabular}{lrrrrc}
			\toprule
			\textbf{Dataset} & \textbf{\#Users} & \textbf{\#Items} & \textbf{\#Inter.}
			                 & \textbf{Density} & \textbf{Shift}                                     \\
			\midrule
			MovieLens(ML)-100K   & 942              & 1{,}447          & 55{,}375          & 4.058\% & T \\
			Coat             & 290              & 300              & 6{,}960           & 8.000\% & E \\
			Amazon-Book      & 10{,}000         & 89{,}911         & 1{,}550{,}000     & 0.172\% & P \\
			Amazon-Beauty    & 22{,}363         & 12{,}101         & 198{,}502         & 0.073\% & P \\
			MovieLens(ML)-1M            & 6{,}040          & 3{,}706          & 1{,}000{,}209     & 4.468\% & T \\
			\bottomrule
		\end{tabular}
\end{table}

\noindent\textbf{Baselines and evaluation metrics.} We compare GenCAR with the standard BPRMF ranker~\cite{rendle2009bpr},
propensity-based methods including CausE~\cite{bonner2018cause},
MACR~\cite{wei2021model}, and IPS-CN~\cite{schnabel2016recommendations},
and representation-based OOD recommenders including
DICE~\cite{zheng2021disentangling},
InvCF~\cite{zhang2023invcf}, and CDR~\cite{wang2023cor}.
We also include CausalDiffRec~\cite{zhao2025causaldiffrec} as a
generative baseline. TallRec~\cite{bao2023tallrec} provides a separate
reference for methods that invoke an LLM during serving. The CausalVAE
backbone isolates the contribution of GenCAR's offline counterfactual
pipeline on ML-100K and Coat. We report Recall@$K$ and normalized discounted cumulative gain (NDCG)@$K$ for
$K\in\{10,20\}$. Recall@10 is the primary metric because GenCAR targets
the recovery of useful candidates near the top of the served list. Inverse propensity score (IPS)-weighted evaluation measures the same utility after correcting for
exposure frequency. Ranking results are averaged over seeds 2024,
2025, and 2026. All methods use the same shift-specific 80/10/10
interaction split and evaluation protocol. Model selection is performed on the validation split, and the reported results are computed on the held-out test
split.

\noindent\textbf{Implementation.} All models are implemented in PyTorch and trained on NVIDIA L20 graphics processing units (GPUs). We use Adam with a learning rate of $10^{-3}$ and weight decay of $10^{-4}$. The CausalVAE encoder produces factorized
posteriors with $\dim(\zc)=64$ and $\dim(\ze)=\dim(\eta)=16$ from a shared 256-dimensional hidden representation. We set $\lambda_{\mathrm{BPR}}=1.0$,
$\lambda_{\mathrm{disent}}=0.05$, and the Kullback--Leibler (KL) weight $\lambda_{\mathrm{KL}}=0.2$. The disentanglement loss is computed with 256 sampled users per update. The \emph{fixkl} configuration stabilizes the preference and environment factors before counterfactual generation.  We use $K_a=5$ preference anchors and set the default trust radius to $\delta=0.7$. DeepSeek-Chat generates the counterfactual candidates, and prompts are cached for each dataset and seed. The LLM is used
only during offline training. Amazon-Book contains sparse long-tail interactions that provide weak item representations. We initialize its item embeddings with BAAI General Embedding (BGE)-small-en text representations projected by principal component analysis. An anchor penalty with $\lambda_{\mathrm{anc}}=5$ keeps the fine-tuned embeddings close to this text prior. The proxy-alignment threshold is $\tau=0.85$, and the reported risk audit uses BH at $\alpha=0.30$. Alignment fitting and calibration use disjoint users, so no user contributes to both predictor fitting and the null reference distribution.

\subsection{Risk Control and Calibration}
\label{sec:exp_conf}
The first experiment audits calibrated selection on five constructed diagnostic pools, separately from the ranker-level candidate-recovery evaluation in Sec.~\ref{sec:exp_main}. Table~\ref{tab:conformal_coverage} evaluates realized proxy FDP and retained support, whereas Sec.~\ref{sec:exp_main} evaluates top-$K$ OOD recovery. Proxy-alignment fitting and null calibration use disjoint users in every dataset, and BH at $\alpha=0.30$ is applied only to the held-out test pool. Table~\ref{tab:conformal_coverage} reports both sides of the resulting selection profile: realized proxy FDP and the amount of retained support.

\begin{table}
\centering
\caption{Realized proxy FDP and proxy-alignment rate under BH selection.}
\label{tab:conformal_coverage}
\small
\setlength{\tabcolsep}{4pt}
\begin{tabular}{lccccc}
\toprule
\textbf{Dataset} & \boldmath$\alpha$ & \textbf{$n_{\mathrm{test}}$}
& \textbf{$n_{\mathrm{sel}}$}
& \boldmath$\widehat{\mathrm{FDP}}$
& \boldmath$1-\widehat{\mathrm{FDP}}$ \\
\midrule
Coat        & 0.30 & 1{,}376   & 16    & 0.188 & 0.812 \\
ML-100K     & 0.30 & 16{,}614  & 241   & 0.195 & 0.805 \\
ML-1M       & 0.30 & 172{,}586 & 3{,}891 & 0.089 & 0.911 \\
Beauty      & 0.30 & 21{,}840  & 1{,}487 & 0.054 & 0.946 \\
Amazon-Book & 0.30 & 559{,}658 & 7{,}665 & 0.048 & 0.952 \\
\bottomrule
\end{tabular}
\end{table}

\noindent\textbf{GenCAR keeps realized proxy FDP below \(\alpha\) on all five candidate pools.} At \(\alpha=0.30\), Table~\ref{tab:conformal_coverage} reports realized proxy FDP values of \(0.188\), \(0.195\), \(0.089\), \(0.054\), and \(0.048\) on Coat, ML-100K, ML-1M, Amazon-Beauty, and
Amazon-Book, respectively. The selector retains between 16 and 7,665 user--item pairs. The corresponding retention rates are \(1.16\%\), \(1.45\%\), \(2.25\%\), \(6.81\%\), and \(1.37\%\), showing how retained support varies with pool scale and calibration support at a fixed risk level.

\noindent\textbf{Positive calibration support tracks BH retention.} The variation in retained support across pools motivates Fig.~\ref{fig:conformal_analysis}(a), which relates the observed BH pass rate to the number of positive calibration pairs $n_{\mathrm{cal}+}$. Across the five pools, larger positive support is associated with higher retention at the same $\alpha$.

\begin{figure}
\centering
\includegraphics[width=\columnwidth]{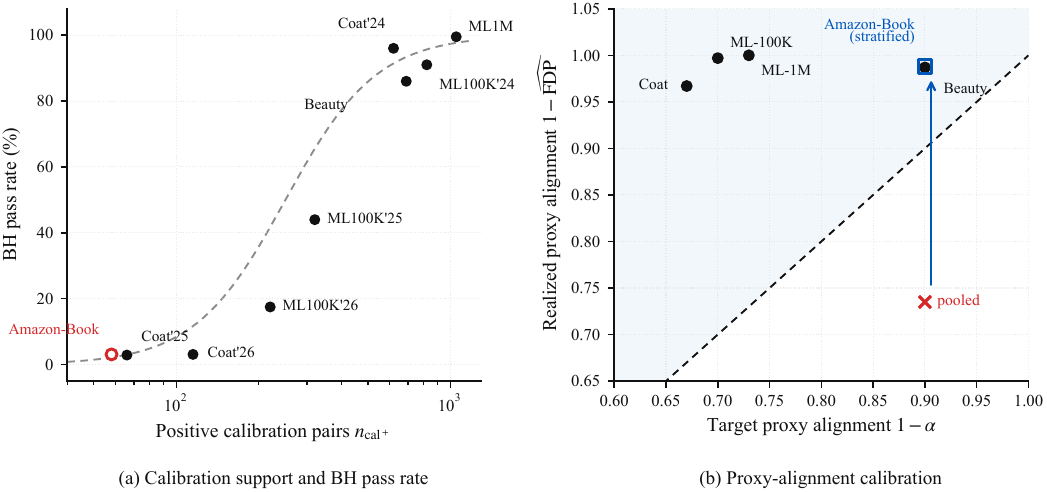}
\caption{Finite-sample calibration diagnostics. (a) BH pass rate as a function of positive calibration support. (b) Empirical proxy alignment relative to the nominal target across runs. The Amazon-Book results compare global and popularity-stratified calibration at $\alpha=0.10$.}
\label{fig:conformal_analysis}
\end{figure}

\noindent\textbf{Popularity stratification aligns heterogeneous score regions.}
Fig.~\ref{fig:conformal_analysis}(b) examines popularity-dependent score heterogeneity without changing the underlying recommendations: global and popularity-stratified calibration are applied to the same Amazon-Book candidate pool at $\alpha=0.10$, with the ranker and candidate pool held fixed. Candidates are grouped by training popularity, calibrated within each group, and pooled after selection. At $\alpha=0.10$, this diagnostic raises empirical proxy alignment from $0.735$ to $0.988$, exceeding the nominal target of $0.90$ with the ranker fixed.

\subsection{OOD Candidate Recovery}
\label{sec:exp_main}

This experiment evaluates whether preference-grounded counterfactual supervision improves OOD candidate recovery under diverse shifts. We compare GenCAR with methods for propensity correction, invariant representation learning, and counterfactual generation. TallRec provides an inference-time LLM reference. Recall@10 is the primary metric because it measures useful candidates near the top of each ranked list. The matched CausalVAE comparison holds the backbone, data splits, and evaluation protocol fixed, leaving the offline counterfactual pipeline as the experimental difference. Table~\ref{tab:main_results} reports the resulting means over three seeds.

\begin{table*}
\centering
\caption{Mean OOD ranking performance over three seeds under temporal,
exposure, and popularity shifts. \textbf{Bold} and \underline{underlined} entries denote the best and second-best results among methods without online LLM
calls. TallRec is included as an inference-time LLM reference. R@$K$ and N@$K$ denote Recall@$K$ and NDCG@$K$, respectively; BB denotes the CausalVAE backbone.}
\label{tab:main_results}
\setlength{\tabcolsep}{3pt}
\small
\resizebox{\textwidth}{!}{%
\begin{tabular}{l cccc cccc cccc}
\toprule
& \multicolumn{4}{c}{\textbf{ML-100K}}
& \multicolumn{4}{c}{\textbf{Coat}}
& \multicolumn{4}{c}{\textbf{Amazon-Book}} \\
\cmidrule(lr){2-5} \cmidrule(lr){6-9} \cmidrule(lr){10-13}
\textbf{Method}
& R@10 & R@20 & N@10 & N@20
& R@10 & R@20 & N@10 & N@20
& R@10 & R@20 & N@10 & N@20 \\
\midrule
BPRMF~\cite{rendle2009bpr}
& 0.0409 & 0.0794 & 0.0939 & 0.0990
& 0.0726 & 0.1238 & 0.0415 & 0.0586
& 0.0002 & 0.0003 & 0.0005 & 0.0005 \\
\midrule
DICE~\cite{zheng2021disentangling}
& 0.0445 & 0.0707 & 0.0929 & 0.0942
& \underline{0.0730} & \textbf{0.1285} & 0.0403 & 0.0587
& 0.0001 & 0.0003 & 0.0005 & 0.0004 \\
IPS-CN~\cite{schnabel2016recommendations}
& 0.0329 & 0.0636 & 0.0699 & 0.0751
& 0.0674 & 0.1175 & 0.0409 & 0.0582
& 0.0001 & 0.0002 & 0.0003 & 0.0003 \\
CausE~\cite{bonner2018cause} 
& 0.0337 & 0.0623 & 0.0841 & 0.0854
& 0.0535 & 0.0956 & 0.0357 & 0.0502
& 0.0001 & 0.0002 & 0.0003 & 0.0003 \\
MACR~\cite{wei2021model}
& 0.0013 & 0.0027 & 0.0032 & 0.0036
& 0.0157 & 0.0317 & 0.0119 & 0.0177
& 0.0001 & 0.0002 & 0.0003 & 0.0003 \\
InvCF~\cite{zhang2023invcf} 
& 0.0322 & 0.0638 & 0.0748 & 0.0774
& 0.0641 & 0.1183 & 0.0374 & 0.0554
& 0.0001 & 0.0002 & 0.0003 & 0.0003 \\
CDR~\cite{wang2023cor}
& 0.0300 & 0.0591 & 0.0676 & 0.0713
& 0.0518 & 0.0825 & 0.0315 & 0.0416
& 0.0001 & 0.0002 & 0.0003 & 0.0003 \\
CausalDiffRec~\cite{zhao2025causaldiffrec}
& 0.0116 & 0.0225 & 0.0422 & 0.0439
& 0.0402 & 0.0596 & 0.0274 & 0.0345
& 0.0002 & 0.0004 & 0.0006 & 0.0006 \\
\midrule
\multicolumn{13}{l}{\emph{LLM-based recommenders} (inference-time LLM call)} \\
TallRec~\cite{bao2023tallrec} (Qwen2.5-3B + LoRA)
& 0.0675 & 0.1231 & 0.1746 & 0.1669
& 0.0637 & 0.1194 & 0.0446 & 0.0606
& 0.0070 & 0.0124 & 0.0188 & 0.0177 \\
\midrule
CausalVAE (BB)
& \underline{0.0856} & \underline{0.1409} & \underline{0.1808} & \textbf{0.1842}
& 0.0639 & \underline{0.1197} & \underline{0.0417} & \underline{0.0603}
& \underline{0.0046} & \underline{0.0086} & \underline{0.0116} & \underline{0.0116} \\
\rowcolor{gray!12}
\textbf{GenCAR}
& \textbf{0.0950} & \textbf{0.1438} & \textbf{0.1820} & \underline{0.1798}
& \textbf{0.0761} & \underline{0.1240} & \textbf{0.0438} & \textbf{0.0604}
& \textbf{0.0066} & \textbf{0.0120}
& \textbf{0.0165} & \textbf{0.0162} \\
$\Delta$ vs BB
& 11.0\%\,$\uparrow$ & 2.1\%\,$\uparrow$ & 0.7\%\,$\uparrow$  & 2.4\%\,$\downarrow$
& 19.1\%\,$\uparrow$ & 3.6\%\,$\uparrow$ & 5.0\%\,$\uparrow$  & 0.2\%\,$\uparrow$
& 43.5\%\,$\uparrow$ & 39.5\%\,$\uparrow$
& 42.2\%\,$\uparrow$ & 39.7\%\,$\uparrow$ \\
\bottomrule
\end{tabular}}
\end{table*}

\noindent\textbf{GenCAR achieves the strongest LLM-free top-10 recovery
across all three shifts.} Table~\ref{tab:main_results} reports means over three seeds. GenCAR ranks first in 10 of the 12 LLM-free dataset--metric comparisons and second in the
remaining two. Its Recall@10 improves over the corresponding CausalVAE
references from $0.0856$ to $0.0950$, from $0.0639$ to $0.0761$, and
from $0.0046$ to $0.0066$, corresponding to gains of $11.0\%$,
$19.1\%$, and $43.5\%$. The gain is positive under all three shifts and
in all nine dataset--seed comparisons
(Fig.~\ref{fig:main_comparison} and Fig.~\ref{fig:sensitivity}(c)).

\noindent\textbf{GenCAR preserves an LLM-free online path.}
The TallRec reference separates offline counterfactual supervision from LLM access during serving: GenCAR uses embedding-based inference, whereas TallRec invokes an LLM online on the same datasets and ranking metrics. GenCAR exceeds the inference-time TallRec reference on ML-100K
(\(0.0950\) versus \(0.0675\)) and Coat (\(0.0761\) versus \(0.0637\)).
On Amazon-Book, GenCAR approaches TallRec
(\(0.0066\) versus \(0.0070\)) while retaining embedding-based online
scoring. The cross-shift recovery gains therefore do not require an
online LLM call.

\begin{figure*}
		\centering
		\includegraphics[width=\textwidth]{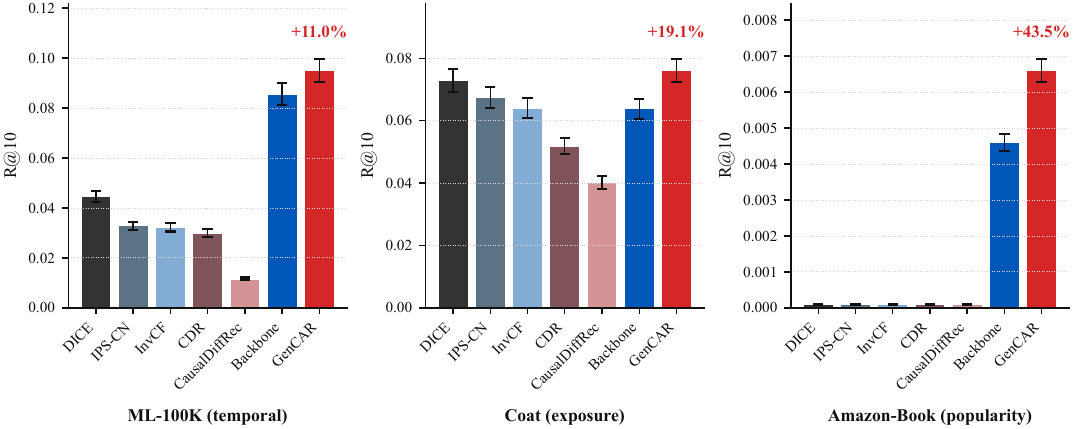}
		\caption{OOD performance under temporal, exposure, and popularity shifts. Bars report means over three seeds. Percentages show GenCAR's relative gains over the corresponding CausalVAE backbone}
		\label{fig:main_comparison}
\end{figure*}

\subsection{Mechanism and Sensitivity Analysis}
\label{sec:exp_ablation}
The preceding results establish GenCAR's cross-shift Recall@10 gain.
We next examine whether this gain arises from counterfactual content,
candidate novelty, trust-radius filtering, or preference factorization.
We vary one factor at a time under matched backbones and evaluation protocols. The resulting comparisons address four distinct explanations for the main gain: augmentation volume, candidate novelty, preference compatibility, and run-to-run variation.

\noindent\textbf{Preference-grounded counterfactuals outperform
matched augmentation controls.}
On Coat, random and inverse-popularity augmentation add the same number of training pairs as the LLM counterfactuals, controlling for augmentation volume. Under the shared ablation configuration, LLM counterfactuals achieve
a Recall@10 of \(0.0720\) on Coat, compared with \(0.0641\) for random
augmentation and \(0.0653\) for inverse-popularity sampling. The
\(12.3\%\) gain over the matched random control isolates the value of
counterfactual content from the number of additional training pairs.

\noindent\textbf{Offline LLM choice separates novelty from ranking gain.} With augmentation volume controlled, candidate novelty remains a separate explanation. Fig.~\ref{fig:llm_choice} compares the novelty profiles of DeepSeek-Chat, V4-Flash, and V4-Pro under the same Coat backbone with their Recall@10 gains over the backbone. DeepSeek-Chat generates $88.6\%$ novel candidates and improves Recall@10 by $8.0\%$. V4-Flash and V4-Pro each reach $100\%$ novelty, with gains of $18.7\%$ and $11.5\%$. Equal novelty with different recovery gains shows that support expansion and candidate utility capture distinct properties.
\begin{figure}
\centering
\includegraphics[width=0.95\columnwidth]{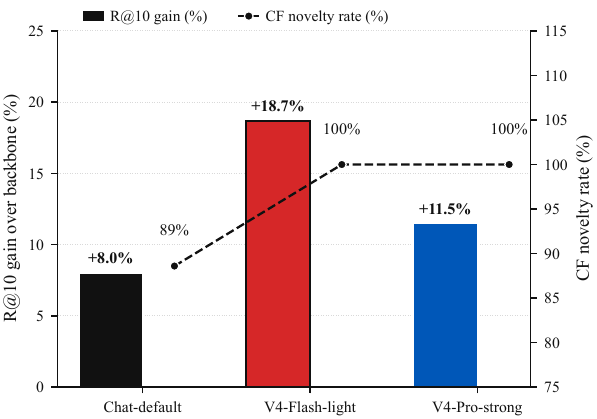}
\caption{Effect of the offline LLM on Coat. Bars report the Recall@10 improvement over the CausalVAE backbone. Circle markers report the fraction of generated candidates absent from the observed training support.}
\label{fig:llm_choice}
\end{figure}

\noindent\textbf{Trust-radius filtering has a stable operating region
on the primary ranking benchmarks.}
Because $\delta$ determines which generated candidates are treated as preference-compatible, we sweep $\delta\in\{0.3,0.5,0.7,0.9,1.0\}$ on ML-100K and Coat with the remaining configuration fixed. Fig.~\ref{fig:sensitivity}(a) shows that ML-100K reaches its
largest gain at \(\delta=0.5\), with gradual variation between
\(\delta=0.5\) and \(0.7\). Coat exhibits a similarly stable region
around the default value \(\delta=0.7\). These curves identify a broad
operating range for filtering preference-compatible proposals.

\begin{figure*}
\centering
\includegraphics[width=\textwidth]{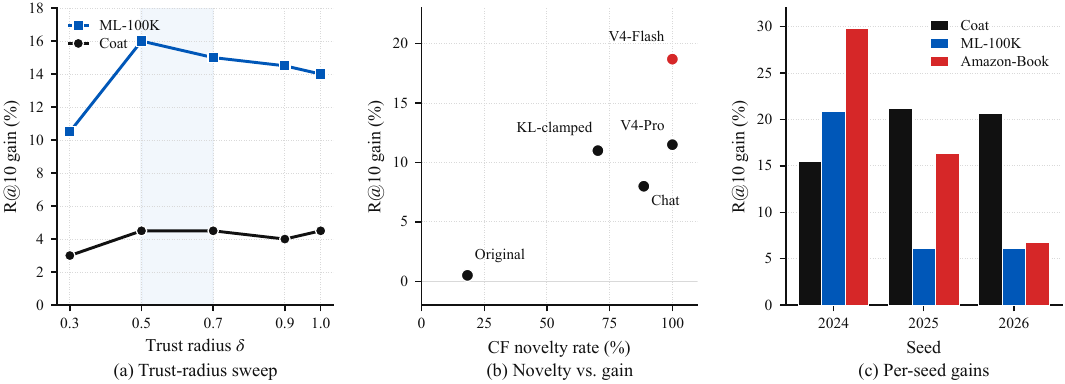}
\caption{Mechanism and sensitivity analysis.
(a) Recall@10 gain as the trust radius varies.
(b) Association between counterfactual novelty and ranking gain.
(c) Per-seed gains under the main GenCAR configuration.}
\label{fig:sensitivity}
\end{figure*}

\noindent\textbf{KL-clamped preference factorization improves novelty and recovery together.} The factorization analysis pairs counterfactual novelty with Recall@10 gain, comparing the original backbone with the KL-clamped (\emph{fixkl}) configuration. Fig.~\ref{fig:sensitivity}(b) shows that novelty rises from $18.3\%$ to $70.3\%$, while the Recall@10 gain rises from near zero to approximately $8.0\%$. This paired change links stable preference factorization to more useful counterfactual supervision. Fig.~\ref{fig:sensitivity}(c) then extends the aggregated comparison to the three individual seeds under each main shift. The figure reports positive Recall@10 gains in all nine dataset--seed comparisons, covering three shifts and three seeds. This consistency supports the repeatability of the main candidate-recovery result.

\subsection{Qualitative Analysis}
\label{sec:exp_case_study}
The aggregate evaluations quantify pooled proxy-label FDP and OOD candidate recovery. We next examine how GenCAR changes candidate composition for individual users. The analysis traces the environmental intervention, candidate-level proxy evidence, and latent candidate placement.

\noindent\textbf{Environmental intervention changes candidate composition under fixed preference.} Fig.~\ref{fig:case_user156} follows Coat User~156 from logged missing-not-at-random (MNAR) exposure to the counterfactual missing completely at random (MCAR) intervention: $\zc$ is fixed, $\ze$ is replaced by the MCAR target, and counterfactual items found among the held-out positives are marked. Bomber jackets account for $50\%$ of logged interactions and $33\%$ of the counterfactual distribution. The intervention assigns $22\%$, $22\%$, and $23\%$ to Packable, Rain, and Other jackets; Packable and Other also occur among the held-out MCAR positives. This case traces the change induced by $\operatorname{do}(\ze:=\zet^{\mathrm{MCAR}})$ while $\zc$ remains fixed.
\begin{figure*}
\centering
\includegraphics[width=\textwidth]
{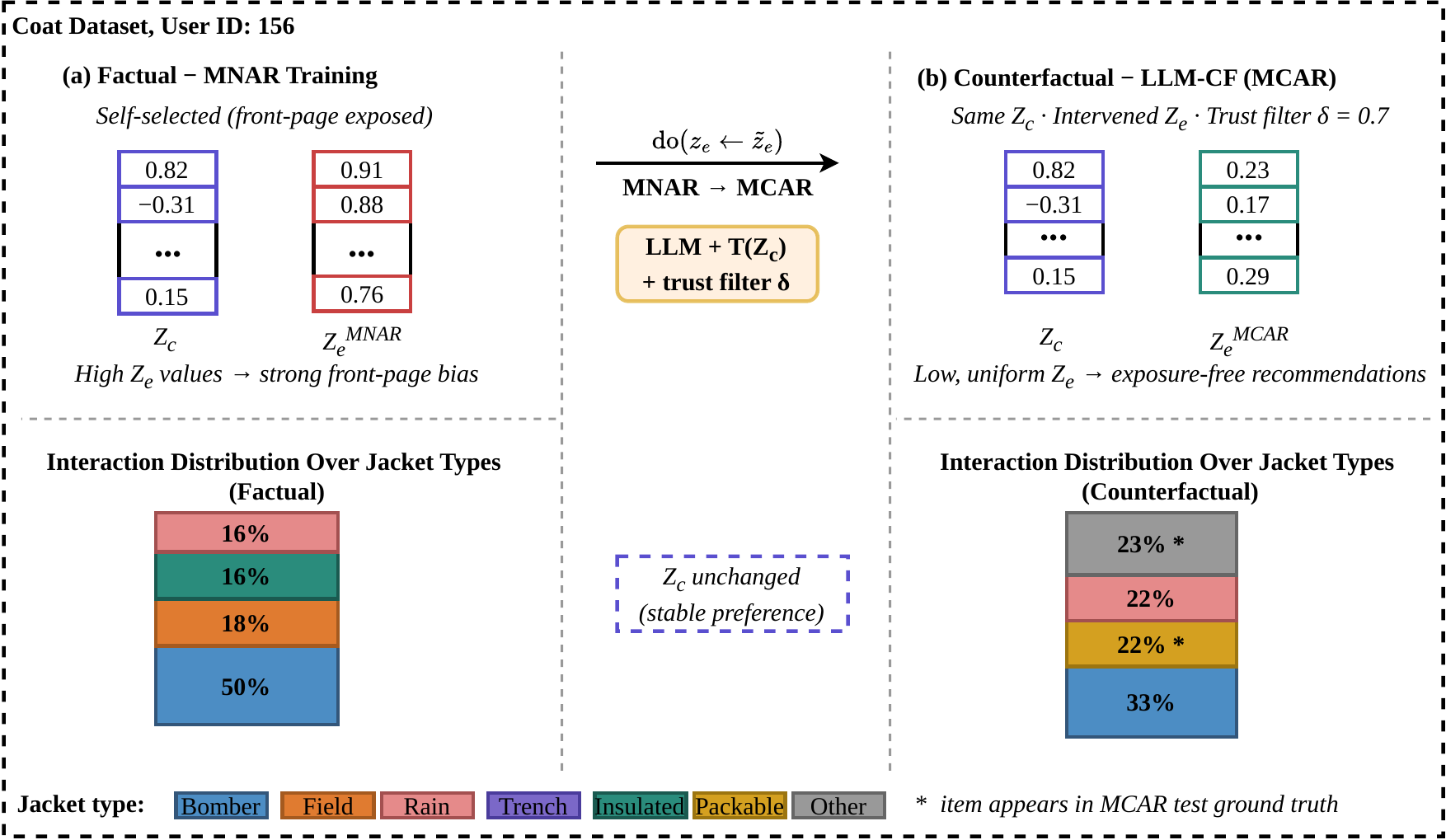}
\caption{Environmental intervention for Coat User~156.
(a) Logged interactions under MNAR exposure.
(b) Counterfactual candidates under
$\mathrm{do}(\ze:=\zet^{\mathrm{MCAR}})$.
The preference representation $\zc$ remains fixed, and filled circles mark
items that occur in the held-out MCAR positives. MNAR and MCAR denote missing not at random and missing completely at random, respectively.}
\label{fig:case_user156}
\end{figure*}

\noindent\textbf{Proxy-alignment scores organize candidate-level evidence before conformal calibration.} Table~\ref{tab:case_study} complements the pooled risk audit with candidate-level evidence from five Amazon-Book users under popularity shift. Alignment scores are shown above, near, and below $\tau=0.85$, with held-out positives marked for reference. For User 818, \emph{The Sense of an Ending} receives a proxy-alignment score of \(0.90\). This item is a held-out positive that is absent from the backbone shortlist. Users 166 and 449 include candidates above, near, and below the proxy threshold. These cases visualize the graded candidate-level evidence available before conformal calibration. The additional users show the same score ordering across different preference histories. For User~974, goal- and mastery-oriented titles receive scores of \(0.85\)--\(0.95\), while the more broadly related \emph{Wired to Care} receives \(0.50\). For User~1249, \emph{Dark Mirror} and \emph{Starship Troopers} receive \(0.90\) and \(0.85\), compared with \(0.30\) for the cross-genre \emph{My Side of the Mountain}. Across these cases, the scores provide graded candidate evidence for the specified counterfactual preference before conformal calibration.
\begin{table*}[h]
\centering
\caption{Amazon-Book counterfactual candidates under popularity shift. \colorbox{Cyan!20}{Blue} denotes candidates with proxy-alignment scores at or above \(\tau=0.85\), \colorbox{gray!15}{gray} denotes candidates near the threshold, and  \colorbox{yellow!20}{yellow} denotes candidates with lower scores. A $\bigstar$\ marks held-out positives.}
\label{tab:case_study}
\footnotesize
\renewcommand{\arraystretch}{1.1}
\setlength{\tabcolsep}{3pt}
\begin{tabular}{c p{3.5cm} p{8.6cm} p{3.5cm}}
\toprule
\textbf{User} & \textbf{Training history}
  & \textbf{LLM counterfactual items (alignment score)}
  & \textbf{Held-out positives} \\
\midrule
U\,818
& The Dying Animal; Kafka on the Shore; Charlotte Simmons ($71$ total)
& \colorbox{Cyan!20}{A Visit from the Goon Squad (0.90)};
  \colorbox{Cyan!20}{\textbf{The Sense of an Ending (0.90)}~$\bigstar$};
  \colorbox{Cyan!20}{Year of Wonders (0.85)};
  \colorbox{gray!15}{The Road (0.80)};
  \colorbox{yellow!20}{Stoner (0.75)}
& \textbf{The Sense of an Ending}~$\bigstar$; How to Live;
  Blue Nights ($16$ total) \\
\midrule
U\,166
& L.A.\ Requiem; Deck the Halls; Gone for Good; Hard Eight ($56$ total)
& \colorbox{Cyan!20}{Tuesdays with Morrie (0.95)};
  \colorbox{Cyan!20}{The Alchemist (0.90)};
  \colorbox{Cyan!20}{I Know Why the Caged Bird Sings (0.85)};
  \colorbox{gray!15}{Harry Potter: Prisoner of Azkaban (0.80)};
  \colorbox{yellow!20}{War and Peace (0.60)}
& When She Was Bad; Brimstone; The Dark Tide ($14$ total) \\
\midrule
U\,449
& Catechism of the Catholic Church; Dreamcatcher; A Simple Plan
  ($56$ total)
& \colorbox{Cyan!20}{A Fine Balance (0.90)};
  \colorbox{Cyan!20}{John Adams (0.85)};
  \colorbox{Cyan!20}{A Simple Plan (0.85)};
  \colorbox{gray!15}{Man's Search For Meaning (0.80)};
  \colorbox{yellow!20}{Atlas Shrugged (0.75)}
& Killing Floor (Jack Reacher); Cell: A Novel;
  No Country for Old Men ($14$ total) \\
\midrule
U\,974
& Atlas Shrugged; Capitalism: The Unknown Ideal; Extreme Programming Explained; Good to Great ($71$ total)
& \colorbox{Cyan!20}{Put Your Dream to the Test (0.95)}; \colorbox{Cyan!20}{Peaks and Valleys (0.90)}; \colorbox{Cyan!20}{Mastery (0.85)}; \colorbox{gray!15}{What Got You Here Won't Get You There (0.82)}; \colorbox{yellow!20}{Wired to Care (0.50)}
& The Case Against the Fed; Buy-In; Your Brain at Work ($18$ total) \\
\midrule
U\,1249
& Starman Jones; Tunnel in the Sky; Imzadi; Dyson Sphere ($56$ total)
& \colorbox{Cyan!20}{Star Trek The Next Generation: Dark Mirror (0.90)}; \colorbox{Cyan!20}{Starship Troopers (0.85)}; \colorbox{gray!15}{Eaters of the Dead (0.80)}; \colorbox{yellow!20}{Ringworld (0.75)}; \colorbox{yellow!20}{My Side of the Mountain (0.30)}
& The Fabulous Riverboat; Nightfall and Other Stories; Job: A Comedy of Justice ($15$ total) \\
\bottomrule
\end{tabular}
\end{table*}

We complement the item-level cases with the latent-space view in Fig.~\ref{fig:disentangle_and_cf}. For panel (a), the two 64-dimensional representations of all Coat users are stacked into a shared $(2n_{\mathrm{users}})\times64$ matrix and projected jointly using t-distributed stochastic neighbor embedding (t-SNE)~\cite{JMLR:v9:vandermaaten08a}, with perplexity 20, 1{,}000 iterations, principal component analysis (PCA) initialization, and random seed 0. Panel (b) overlays the training interactions, backbone top-10, six counterfactual candidates, and held-out positives for User~240 on the 300-item map.

 Fig.~\ref{fig:disentangle_and_cf}(a) places the preference and environment representations in separated regions of the two-dimensional projection, which is consistent with the factorization used for counterfactual generation. Fig.~\ref{fig:disentangle_and_cf}(b) traces Coat User~240. One of six counterfactual candidates matches a held-out positive, and several lie outside the backbone's top-10 neighborhood. The generated set therefore expands this user's shortlist with a relevant held-out item.
 
\begin{figure*}
		\centering
		\includegraphics[width=\textwidth]{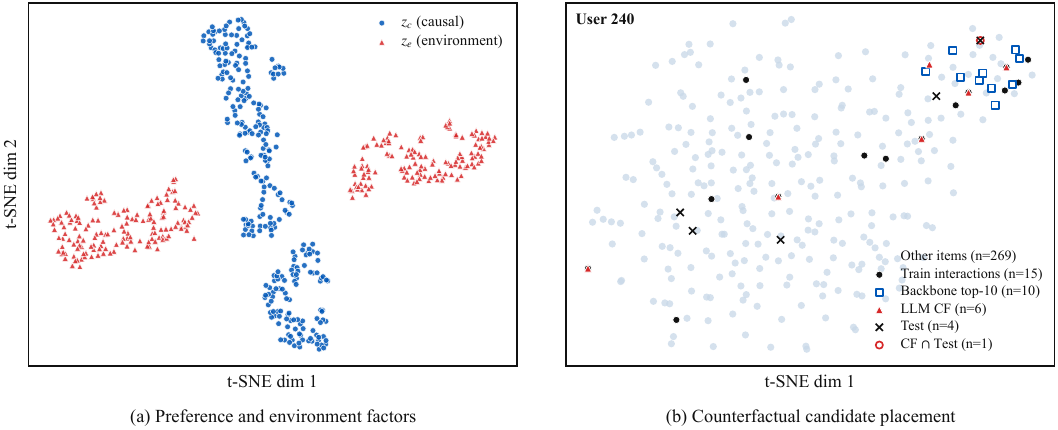}
		\caption{Latent geometry and counterfactual candidate placement on Coat. (a) Two-dimensional t-distributed stochastic neighbor embedding (t-SNE) projection of the preference and environment representations. (b) Candidate neighborhood for User 240, including generated candidates, backbone top-10 items, and held-out positives. CF denotes counterfactual.}
		\label{fig:disentangle_and_cf}
\end{figure*}

\subsection{Serving Efficiency and Partial Offline Coverage}
\label{sec:cost}
Deployment depends on both offline coverage and online cost, which we report separately in Tables~\ref{tab:sequential_results} and~\ref{tab:cost}. The partial-coverage study evaluates all 2,664 Steam users, although offline LLM counterfactuals are available for only 1,020 of them. The comparison includes the CausalVAE+BGE backbone, anchor counterfactual fine-tuning, and an auxiliary reranking variant that explicitly retains available counterfactual items.

\noindent\textbf{GenCAR retains ranking utility under partial offline
coverage.}
Table~\ref{tab:sequential_results} evaluates 2,664 Steam users, of whom
1,020 receive offline LLM counterfactuals. With $38\%$ offline coverage, anchor counterfactual fine-tuning remains close to the CausalVAE+BGE backbone ($0.0409$ versus $0.0411$ in Recall@10). The auxiliary GenCAR + rerank variant explicitly retains available counterfactual items in the ranked list, raising Recall@10 to $0.0449$ and NDCG@10 to $0.0725$, gains of $9.2\%$ and $6.9\%$. This variant measures partial-coverage utility; the standard pooled selector remains the certified output.

\begin{table}[t]
\centering
\caption{Ranking performance on Steam with offline LLM counterfactuals available for \(38\%\) of the evaluated users. GenCAR + rerank is the auxiliary explicit-retention variant. R@$K$ and N@$K$ denote Recall@$K$ and NDCG@$K$, respectively; CF denotes counterfactual.}
\label{tab:sequential_results}
\footnotesize
\setlength{\tabcolsep}{3pt}
\begin{tabular}{lcccc}
\toprule
\textbf{Method} & \textbf{R@10} & \textbf{R@20} & \textbf{N@10} & \textbf{N@20} \\
\midrule
CausalVAE + BGE backbone & 0.0411 & 0.0646 & 0.0678 & 0.0690 \\
GenCAR anchor CF fine-tune & 0.0409 & 0.0628 & 0.0677 & 0.0686 \\
\textbf{GenCAR + rerank} & \textbf{0.0449} & \textbf{0.0661}
& \textbf{0.0725} & \textbf{0.0719} \\
\bottomrule
\end{tabular}
\end{table}

\noindent\textbf{Offline counterfactual generation preserves sub-millisecond online serving.} Table~\ref{tab:cost} reports trainable parameters, training time per seed, per-user inference latency, and online LLM usage on Coat. GenCAR's one-time generation stage is treated as an offline cost, while TallRec's inference-time LLM call is included in online latency. GenCAR and its CausalVAE backbone each contain \(0.15\) million trainable parameters and require less than one millisecond per user. TallRec requires approximately \(5.4\) seconds per user in the reported setup. GenCAR's online path consists of user encoding, inner-product scoring, and calibrated selection. Thus, we add counterfactual supervision while retaining the backbone’s model size and sub-millisecond serving profile.
\begin{table}[t]
\centering
\caption{Model size and online serving cost on Coat. One-time offline LLM counterfactual generation of GenCAR is excluded from inference latency. BB denotes the CausalVAE backbone; M denotes millions of parameters; CF denotes counterfactual.}
\label{tab:cost}
\small
\setlength{\tabcolsep}{1pt}
\renewcommand{\arraystretch}{1}
\begin{tabular}{lcccc}
\toprule
\multirow{2}{*}{\textbf{Method}} & \textbf{Params}& \textbf{Training}& \textbf{Inference} &\textbf{LLM}\\
& (M)& (min/seed)& (ms/user) & \textbf{Usage} \\
\midrule
BPRMF          & 0.04       & 2         & $<\!1$           & $\times$ \\
DICE/InvCF/CDR & 0.08--0.12 & 8--12     & $<\!1$           & $\times$ \\
CausalDiffRec  & 0.85       & 53        & ${\sim}35^{*}$   & $\times$ \\
TallRec        & ${\sim}3000$ & ${\sim}30$ & ${\sim}5400^\dagger$ & infer \\
\midrule
CausalVAE (BB) & 0.15       & 30        & $<\!1$           & $\times$ \\
\rowcolor{gray!12}
\textbf{GenCAR} & \textbf{0.15} & \textbf{30 + offline$^{\ddagger}$} & \textbf{$<1$} & Offline\\
\bottomrule
\end{tabular}
\vspace{1mm}
\scriptsize
\noindent$^{*}$50 denoising steps. 
$^{\dagger}$LLM inference cost.
$^{\ddagger}$One-time offline CF.
\end{table}

Together, these results define the deployment profile of GenCAR. Offline LLM counterfactuals can cover a subset of users, and online serving remains LLM-free and sub-millisecond.

\section{Related Work}
\label{sec:related}
Our work connects two lines of research: counterfactual candidate construction for OOD recommendation and predictive inference for risk-controlled selection. Table~\ref{tab:positioning} compares the relevant method families in terms of OOD scope, candidate construction, certified target, and online LLM usage.

\noindent\textbf{OOD and counterfactual recommendation.}
Out-of-distribution recommendation seeks to preserve ranking utility when interaction distributions change between training and serving. Reviews organize causal recommendation through potential-outcome, structural causal model, and counterfactual formulations~\cite{luo2024causalrec,wang2023causaltutorial}. Causal disentanglement and invariant learning separate preference-related factors from environmental correlations~\cite{zheng2021disentangling,zhang2023invcf,wang2023cor} \cite{yoo2025tpab,liao2025idea}. Distributionally robust objectives and causal diffusion extend this direction to broader forms of distribution shift~\cite{wang2024distributionally,zhao2025causaldiffrec}. Contrastive learning with uniform data further reduces exposure bias in learned representations~\cite{yang2024cdlrec}. These methods primarily reweight observed interactions or learn shift-robust representations. Item-level supervision for a specified environmental intervention remains outside their main objective.

Generative recommenders model complex interaction distributions through diffusion and preference-aware objectives~\cite{wang2023diffrec,yang2023dreamrec,li2023diffurec,liu2025preferdiff} \cite{rajput2023tiger,ma2024pdrec,zhai2024hstu,cai2025diffdiv}. A recent review places these methods within the broader family of generative recommender systems~\cite{deldjoo2024genrecsys}. Surveys of LLM-enhanced recommendation organize language models by their roles in feature construction, representation learning, ranking, and generation~\cite{lin2025llmsurvey,zhao2024llmera}. Individual systems integrate collaborative semantics~\cite{zhang2023collm,zheng2024lcrec,liu2024cccf}, external knowledge~\cite{xi2024kar,zhao2024ikgrec}, or LLM-based ranking~\cite{liao2024llara,hou2024llmrank}, further illustrating the role of language knowledge in recommendation~\cite{sheng2025alpharec} \cite{wei2024llmrec,kim2024allmrec,lin2024rella}. Their primary objectives are representation enrichment, ranking, or direct generation. Intervention-specific supervision constrained by stable preference remains a separate requirement. GenCAR uses an offline LLM to construct item-level counterfactual pairs under an explicit environmental intervention. Stable-preference anchors and trust-radius filtering ground these proposals before ranker training.

\noindent\textbf{Risk-controlled predictive inference.}
Predictive inference quantifies uncertainty and provides finite-sample guarantees for learned outputs~\cite{prinster2024conformal}. Conformal prediction constructs distribution-free prediction sets~\cite{vovk2005,angelopoulos2021gentle}. Conformal risk control extends calibration to the expected values of monotone losses~\cite{angelopoulos2024crc}, while conformal \(p\)-values support selection with FDR control~\cite{bates2023testing,jin2023selection}.

Recommendation-specific methods apply these guarantees to set selection. Angelopoulos et al.~\cite{angelopoulos2023reliable} transform pretrained ranker outputs into recommendation sets with finite-sample FDR control. Conformal Alignment~\cite{gui2024conformal} selects foundation-model outputs using a user-defined alignment criterion. De Toni et al.~\cite{detoni2025flowers} bound unwanted recommendations and expand selected sets using observed user feedback. Their guarantees apply to candidate families paired with observed relevance, alignment, or unwanted-content labels. GenCAR introduces intervention-specific supervision before ranking and then calibrates the ranked candidate family under proxy labels. The \(\alpha\)-VCR formulation makes pooled proxy-label FDR the explicit serving constraint, connecting preference-grounded candidate recovery with calibrated set selection.
\begin{table*}
\centering
\caption{\textbf{Positioning of GenCAR across OOD candidate
construction and risk-controlled serving.}
A dash indicates no finite-sample guarantee for the served set.
``Partial'' means that only some methods in the family explicitly
address distribution shift.}
\label{tab:positioning}
\footnotesize
\setlength{\tabcolsep}{4pt}
\renewcommand{\arraystretch}{1.18}

\begin{tabular}{
@{}
>{\raggedright\arraybackslash}p{0.29\textwidth}
>{\centering\arraybackslash}p{0.06\textwidth}
>{\raggedright\arraybackslash}p{0.25\textwidth}
>{\raggedright\arraybackslash}p{0.21\textwidth}
>{\centering\arraybackslash}p{0.08\textwidth}
@{}
}
\toprule
\textbf{Family}
& \textbf{OOD}
& \textbf{Candidate construction}
& \textbf{Certified target}
& \shortstack{\textbf{Online LLM}} \\
\midrule

Invariant and robust OOD recommendation
~\cite{zhang2023invcf,wang2023cor,
wang2024distributionally,zhao2025causaldiffrec} &Yes
&Observed interactions and invariant representations &\textemdash
&No\\
\addlinespace[2pt]
Diffusion recommendation~\cite{wang2023diffrec,yang2023dreamrec,
li2023diffurec,liu2025preferdiff}&Partial&Diffusion-generated interactions or sequences&\textemdash&No\\
\addlinespace[2pt]
LLM-augmented recommendation~\cite{liao2024llara,xi2024kar,sheng2025alpharec}
&No &LLM knowledge or language representations &\textemdash & Varies\\

\addlinespace[2pt]
Reliable recommendation sets
~\cite{angelopoulos2023reliable}
&
No
&
Fixed ranked candidates
&
Relevance FDR
&
No
\\

\addlinespace[2pt]
Conformal Alignment
~\cite{gui2024conformal}
&
No
&
Foundation-model outputs
&
Output-alignment FDR
&
Model-dependent
\\

\addlinespace[2pt]
Conformal risk control for unwanted recommendations
~\cite{detoni2025flowers}
&
No
&
Ranked and previously consumed items
&
Expected unwanted fraction
&
No
\\

\midrule
\textbf{GenCAR (ours)}
&
\textbf{Yes}
&
\textbf{Ranked candidates learned from environment-intervened,
preference-grounded supervision}
&
\textbf{Pooled proxy-label FDR}
&
\textbf{No}
\\

\bottomrule
\end{tabular}
\end{table*}

\section{Conclusion}
\label{sec:conclusion}
In this paper, we presented GenCAR, a two-stage framework for preference-grounded counterfactual construction and risk-controlled selection in OOD recommendation. By formulating serving as the $\alpha$-Valid Counterfactual Recommendation ($\alpha$-VCR) problem, GenCAR connects environment-intervened candidate construction with FDR-calibrated selection. Theoretically, we bounded conditional counterfactual approximation error, established finite-sample proxy-label FDR control under target-environment exchangeability and PRDS, and characterized the remaining discrepancy between proxy-label and causal risk. Empirically, realized proxy FDP remained below $\alpha$ across five evaluated pools, and GenCAR improved top-10 OOD candidate recovery over the corresponding CausalVAE backbones under diverse shifts. Moving LLM-based counterfactual generation and proxy labeling offline preserves a sub-millisecond, LLM-free online path. GenCAR places counterfactual OOD serving on an explicit statistical footing, with guarantees scoped to proxy labels and their transfer assumptions.

\section*{CRediT authorship contribution statement}
\textbf{Qianqian Wang:} Conceptualization, Methodology, Software, Validation, Formal analysis, Investigation, Writing -- original draft. 
\textbf{Yunshan Li:} Writing – review \& editing.
\textbf{Jiawen Zeng:} Software, Validation.
\textbf{Wenwu Gong:} Writing – review \& editing,
Methodology, Conceptualization.
\textbf{Lili Yang:} Supervision,
Resources, Writing – review \& editing.

\section*{Declaration of competing interest}
The authors declare that they have no known competing financial interests or personal relationships that could have appeared to influence the work reported in this paper.

\section*{Acknowledgements}
This research was funded by the SUSTech Presidential Postdoctoral Fellowship, the China Postdoctoral Science Foundation (Grant No. 2025M773057), Shenzhen Science and Technology Program (Grant No. ZDSYS20210623092007023), and the Shenzhen Key Laboratory of Safety and Security for Next Generation of Industrial Internet, Southern University of Science and Technology.

\section*{Data availability}
The MovieLens, Coat, Amazon-Book, Amazon-Beauty, and Steam datasets used in this study are publicly available from the sources cited in Section~\ref{sec:exp_setup}. Dataset sources, access instructions, and licensing restrictions are documented at \url{https://github.com/nikewq/GenCAR-release}. The model code, training and evaluation pipelines, trained checkpoints, processed evaluation splits, and complete result artifacts will be released in the same repository after manuscript acceptance.

\section*{Declaration of generative AI use}
The authors used generative AI tools only for language polishing and readability improvement.  All technical ideas, experiments, analyses, and conclusions were developed and verified by the authors, who take full responsibility for the content of this paper.

\appendix
\numberwithin{figure}{section}
\numberwithin{table}{section}
\section{Proofs of Theoretical Results}
\label{app:proofs}

\subsection{Proof of Proposition~\ref{thm:cf_validity}}
\label{app:proof_cf_validity}

\begin{proof}
The triangle inequality for total variation gives
\begin{equation*}
\mathrm{TV}
(P^\star,P^M)
\leq
\epsilon_A+\epsilon_M.
\end{equation*}
Let $P=P^\star$, $Q=P^M$, and $A=\mathcal T_u(\delta)$. For any $B\subseteq\mathcal I$,
\begin{align*}
&
\left|
P(B\mid A)-Q(B\mid A)
\right|
\nonumber\\
&\leq
\frac{
|P(B\cap A)-Q(B\cap A)|
}{
P(A)
}
\nonumber\\
&\quad+
Q(B\cap A)
\left|
\frac{1}{P(A)}-\frac{1}{Q(A)}
\right|
\nonumber\\
&\leq
\frac{
2\,\mathrm{TV}(P,Q)
}{
\rho_\delta
}.
\end{align*}
Taking the supremum over $B$ and combining the two bounds proves
the first claim of Proposition~\ref{thm:cf_validity}.

For the event $\mathcal M_u$,
\begin{align*}
P_\delta^M(\mathcal M_u)
&\leq
P_\delta^\star(\mathcal M_u)
+
\left|
P_\delta^M(\mathcal M_u)
-
P_\delta^\star(\mathcal M_u)
\right|
\nonumber\\
&\leq
P_\delta^\star(\mathcal M_u)
+
\frac{2(\epsilon_A+\epsilon_M)}
{\rho_\delta}.
\end{align*}
This proves the second claim of
Proposition~\ref{thm:cf_validity}.
\end{proof}

\subsection{Proof of Lemma~\ref{lem:bh_nested}}
\label{app:proof_bh_nested}

\begin{proof}
By construction, $\mathcal S_{k+1}$ contains all pairs in
$\mathcal S_k$ and one additional pair with the next smallest
$p$-value. Hence
\[
\mathcal S_k\subseteq\mathcal S_{k+1}
\]
for every $k\in\{0,\ldots,N-1\}$, which proves
the nestedness claim of Lemma~\ref{lem:bh_nested}.

The index $\widehat k_\alpha$ is defined as the largest index
satisfying the BH inequality. Thus,
$\mathcal S_{\widehat k_\alpha}$ is the largest member of the nested
family satisfying the step-up condition.

Let $\alpha_1\leq\alpha_2$. Every index satisfying
$p_{(k)}\leq \frac{\alpha_1 k}{N}$ also satisfies
$p_{(k)}\leq\frac{\alpha_2 k}{N}$.
It follows that
$\widehat k_{\alpha_1}\leq\widehat k_{\alpha_2}$.
Nestedness of $\{\mathcal S_k\}_{k=0}^{N}$ gives
$\mathcal S_{\alpha_1}\subseteq\mathcal S_{\alpha_2}$.
\end{proof}

\subsection{Proof of Lemma~\ref{lem:pvalue_validity}}
\label{app:proof_pvalue_validity}

\begin{proof}
Fix a proxy-null candidate $(u,j)$ and condition on the learned
artifacts and candidate family. Let $R^{+}$ be the upper rank of
$V(u,j)$ among the $n_0$ calibration scores and the test score. The
weak-rank $p$-value in Eq.~\eqref{eq:method_pvalue} satisfies
\[
p(u,j)=\frac{R^{+}}{n_0+1}.
\]
Conditional on the multiset of the $n_0+1$ exchangeable scores, at
most $r$ indices have upper rank no greater than $r$. Exchangeability
therefore gives
\begin{equation*}
\Pr\!\left(R^{+}\leq r\mid
\hat h,\mathcal Q(\mathcal B)\right)
\leq \frac{r}{n_0+1}
\end{equation*}
for every $r\in\{0,\ldots,n_0+1\}$. Setting
$r=\lfloor t(n_0+1)\rfloor$ yields
\begin{equation*}
\Pr\!\left(p(u,j)\leq t\mid
\hat h,\mathcal Q(\mathcal B)\right)
\leq t.
\end{equation*}
Thus the implemented weak-rank $p$-value is super-uniform. Ties can
only increase the upper rank and hence preserve the bound.
\end{proof}

\subsection{Proof of Theorem~\ref{thm:fdr} and
Corollary~\ref{cor:by_fdr}}
\label{app:proof_fdr}

\begin{proof}
Condition on the learned artifacts and the candidate family
$\mathcal Q(\mathcal B)$. Lemma~\ref{lem:pvalue_validity}
establishes super-uniformity for every proxy-null $p$-value. Under
conditional PRDS on the proxy-null hypotheses, the BH result of
Benjamini and Yekutieli~\cite{benjamini2001control} gives
\begin{equation*}
\mathbb E\!\left[
\frac{
|\mathcal S_\alpha(\mathcal B)
\cap\mathcal H_0^{\mathrm{proxy}}|
}{
\max\{|\mathcal S_\alpha(\mathcal B)|,1\}
}
\,\middle|\,
\hat h,\mathcal Q(\mathcal B)
\right]
\leq \frac{m_0}{N}\alpha
\leq \alpha.
\end{equation*}
Taking expectations over the conditioning variables and applying the
tower property gives
\begin{equation*}
\mathrm{FDR}_{\mathrm{proxy}}
\bigl(\mathcal S_\alpha(\mathcal B)\bigr)
\leq \alpha,
\end{equation*}
which proves Theorem~\ref{thm:fdr}.

Under arbitrary dependence, the Benjamini--Yekutieli procedure
replaces $\alpha$ with $\alpha/H_N$. Applying its general-dependence
result conditionally and then using the tower property yields
\begin{equation*}
\mathrm{FDR}_{\mathrm{proxy}}
\bigl(\mathcal S_{\alpha}^{\mathrm{BY}}(\mathcal B)\bigr)
\leq \alpha.
\end{equation*}
This proves Corollary~\ref{cor:by_fdr}.
\end{proof}

\subsection{Proof of Corollary~\ref{cor:proxy_transfer}}
\label{app:proof_proxy_transfer}

\begin{proof}
For any selected set $\mathcal S$,
\begin{align*}
|\mathcal S\cap\mathcal H_0^{\mathrm{causal}}|
\leq\;&
|\mathcal S\cap\mathcal H_0^{\mathrm{proxy}}|
\nonumber\\
&+
|\mathcal S\cap
(\mathcal H_0^{\mathrm{causal}}
\setminus\mathcal H_0^{\mathrm{proxy}})|.
\end{align*}
Dividing both sides by $\max\{|\mathcal S|,1\}$ and taking expectations
gives
\begin{equation*}
\mathrm{FDR}_{\mathrm{causal}}(\mathcal S)
\leq
\mathrm{FDR}_{\mathrm{proxy}}(\mathcal S)
+
\Delta_{\mathrm{proxy}}(\mathcal S).
\end{equation*}
Applying Theorem~\ref{thm:fdr} with
$\mathcal S=\mathcal S_\alpha(\mathcal B)$ proves
Corollary~\ref{cor:proxy_transfer}.

If
$\mathcal H_0^{\mathrm{causal}}
\subseteq
\mathcal H_0^{\mathrm{proxy}}$,
then
$\mathcal H_0^{\mathrm{causal}}
\setminus
\mathcal H_0^{\mathrm{proxy}}
=
\emptyset$.
Thus,
$\Delta_{\mathrm{proxy}}
(\mathcal S_\alpha(\mathcal B))=0$,
which proves the final statement.
\end{proof}

\FloatBarrier
\bibliographystyle{elsarticle-num-names}
\bibliography{refs}

\end{document}